\documentclass{article}

\usepackage[letterpaper]{geometry}

\usepackage{graphicx} 
\usepackage{hyperref}
\usepackage{booktabs}
\usepackage{placeins}
\usepackage{algorithm}
\usepackage{algpseudocode}
\usepackage{amsmath,amssymb, amsthm}
\usepackage{amsfonts}
\usepackage{cleveref}
\usepackage{xcolor}
\usepackage[labelfont=bf]{caption}
\usepackage{subcaption}
\usepackage{multirow, enumitem}

\usepackage[switch]{lineno}

\newtheorem{lemma}{Lemma}[section]

\newtheorem{property}{Property}[section]

\newtheorem{definition}{Definition}[section]

\newcommand{\kmeans}{\textsc{$k$-means}}
\newcommand{\weight}{w}
\newcommand{\R}{\mathbb{R}}
\newcommand{\calS}{\mathcal{S}}
\newcommand{\calO}{\mathcal{O}}
\newcommand{\eps}{\varepsilon}
\newcommand{\coreset}{C}
\newcommand{\kmpp}{\textsc{$k$-means\texttt{++}}}
\newcommand{\cost}{\textsc{cost}}

\DeclareMathOperator*{\dist}{dist}
\DeclareMathOperator*{\polylog}{polylog}

\usepackage[colorinlistoftodos,prependcaption,textsize=tiny]{todonotes}

\usepackage{wrapfig}
\newcommand{\erclogowrapped}[1]{%
\setlength\intextsep{0pt}%
\begin{wrapfigure}[3]{r}{#1*\real{1.1}}%
\includegraphics[width=#1]{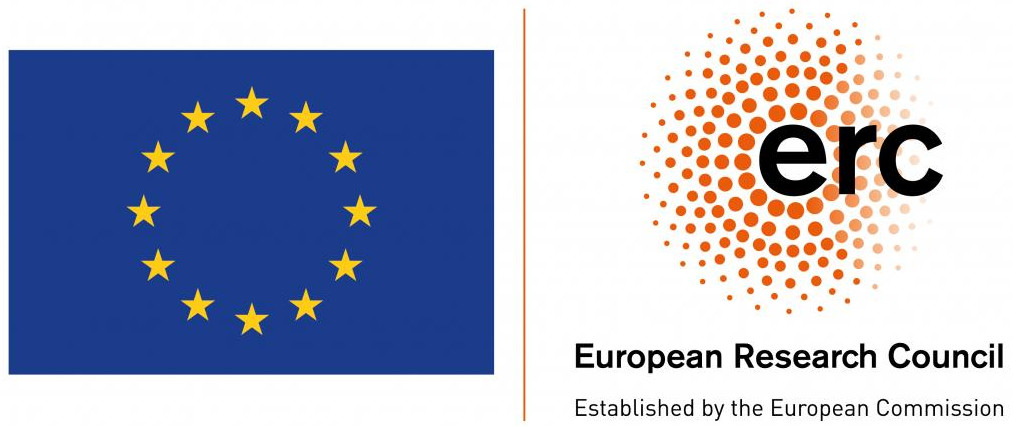}%
\end{wrapfigure}%
}

\title{Experimental Evaluation of Fully Dynamic $k$-Means via Coresets}

\author{Monika Henzinger \thanks{Institute of Science and Technology Austria (ISTA), Klosterneuburg, Austria}
\and David Saulpic\footnotemark[1]
\and Leonhard Sidl \thanks{University of Vienna, Austria. Work done while at ISTA.}}

\date{}

\begin{document}

\maketitle

\begin{abstract}
    For a set of points in $\R^d$, the Euclidean $k$-means problems consists of finding $k$ centers such that the sum of distances squared from each data point to its closest center is minimized. 
    Coresets are one the main tools developed recently to solve this problem in a big data context. They allow to compress the initial dataset while preserving its structure: running any algorithm on the coreset provides a guarantee almost equivalent to running it on the full data.
    In this work, we study coresets in a fully-dynamic setting: points are added and deleted with the goal to efficiently maintain a coreset with which a \kmeans~solution can be computed. 
    Based on an algorithm from Henzinger and Kale [ESA'20], we present an efficient and practical implementation of a fully dynamic coreset algorithm, that improves the running time by up to a factor of 20 compared to our non-optimized implementation of the algorithm by Henzinger and Kale, without sacrificing more than 7\% on the quality of the \kmeans~solution.
    
\end{abstract}

\section{Introduction}

One of the most fundamental tools of data analysis is clustering, where groups of points that are similar, or "close" in some metric space are identified. 
This task can be formalised using the $k$-means problem: find $k$ centers, such that the sum of distances squared from each point to its closest center is minimized. For a given set of centers, this sum is called the \emph{cost} of the solution. A solution with low cost can be used to represent the original data, with only a small loss in precision, by replacing each point with its closest center. A solution also results in a valid clustering of the instance, where a cluster is defined as the set of all points assigned to the same center.

The \kmeans~problem has numerous applications, especially for instances with a very large number of input points \cite{ikotun2022k, bahmani2012scalable}. Various algorithms to solve the problem with different emphasis -- running in near-linear time, as \cite{Cohen-AddadLNSS20}, or using minimal few memory, as \cite{BravermanFLSY17} -- have been described. 

The recently most studied strategy to find a \kmeans~solution on big data problems is to use \emph{coresets}:  for a parameter $\eps > 0$, a set \coreset~is an $\eps$-coreset of the input $P$ if evaluating the cost of any candidate solution on $C$ gives the same cost as evaluating it on the full input, up to a factor of $(1\pm \eps)$.

This line of work developed coresets with a size \emph{independent of the original dataset}, which can be computed equally fast as a solution to \kmeans \cite{Cohen-AddadLSSS22, huang23}.
Additionally, coresets are particularly useful when the available memory is restricted, such as in streaming or distributed settings. We add to the versatility of coresets by focusing on maintaining a coreset in the \emph{fully dynamic} model, where points can be inserted or deleted from the dataset at every time step.

Algorithms specifically designed for the dynamic setting are beneficial in any applications, e.g. social networks or web-search queries, where the dataset is evolving over time, and one wishes to maintain a good clustering. The naive alternative to using a dynamic algorithm is to run a static algorithm after each update; however this can become inefficient, and the solution may even be outdated as soon as it is computed.
To cope with this, Henzinger and Kale~\cite{henzinger2020} presented an algorithm that can efficiently maintain an $\eps$-coreset. More precisely, the complexity to deal with any update is $ O(k^2 d \polylog(n))$, and the $\eps$-coreset obtained has optimal size (essentially $O(k \eps^{-4})$).

Once a coreset can be maintained, it also becomes very efficient to update a solution to \kmeans: For instance, running the popular \kmpp~algorithm of \cite{ArthurV07} would result in a $O(\log k)$ approximation, with a complexity at every time step of $O(k^2 d \polylog(n))$. Alternatively, algorithms that provide a constant factor approximation with the same running time can be used (see e.g. \cite{mettu2004optimal, lattanzi2019better}).

\paragraph{Our Contribution.}
In this work, we explore the practical characteristics of the algorithm described in \cite{henzinger2020} for maintaining both a coreset and a $k$-means solution.\footnote{\url{https://git.ista.ac.at/lsidl/dynamiccoreset} }
We evaluate the practical relevance of several algorithmic details required for strict theoretical guarantees. Furthermore, we propose and test several, partly heuristic, optimizations that improve the theoretical running time of the dynamic algorithm further. After calibrating all parameters on small-scale datasets from \cite{Birchsets}, we test our optimized algorithm on large-scale real-world datasets. Even though our algorithm is in theory versatile enough to be used in any metric space, we focus on euclidean datasets, as they are most common.

We found that an implementation that closely follows the algorithm of \cite{henzinger2020} already vastly outperforms our baseline, despite its heavy details required for theoretical guarantees. 
Combined with our optimizations, this results in an algorithm running more than $550$ times faster on real-world dataset than recomputing from scratch -- which is the only baseline algorithm we are aware of.

For the incremental setting, where all operations are insertions, our optimizations lead to an algorithm provably maintaining a coreset with running time $O(kd \log(n))$ per insertion, where $n$ is the size of the dataset. In the opposite decremental setting, with only deletions, our algorithm also runs in time $O(kd \log n)$, but without a theoretical guarantee on its quality. 
Combining those two guarantees, our algorithm maintains a coreset under both insertion and deletion in amortized time $O(kd \polylog n)$, improving  \cite{henzinger2020} by a factor $k$.

\subsection{Related Work}

\paragraph{On Fully-Dynamic Clustering}
The study of the $k$-means problem dates back to Lloyd's work~\cite{lloyd1982least} on quantization of a continuous signal. Since then, a lot of theoretical research was dedicated to understanding the cost function. In summary, the problem is NP-hard, even when $d=2$ \cite{MahajanNV12} or $k=2$ \cite{dasgupta2009random}. However, it is possible to compute a $(1+\eps)$-approximation \cite{CohenAddadFS19} in low dimensional space, and a $5.91$-approximation in general spaces \cite{Cohen-AddadEMN22}. An orthogonal line of work, more focused on the practical aspect, attempts to find solutions that are as good as possible while minimizing running time: the most celebrated result -- which had a large influence on the practical success of \kmeans~--  is the \kmpp~algorithm of \cite{ArthurV07}, a very simple algorithm that computes a $O(\log k)$-approximation in time $O(nkd)$. 

The fully dynamic setting received lots of attention in recent years, with works of \cite{LattanziV17, BravermanFLSY17, FichtenbergerLN21} for $k$-means, but also \cite{Cohen-AddadHPSS19, BhattacharyaLP22} for the related facility location problem or \cite{ChanLSW22, BateniEFHJMW23} for $k$-center. 
However, 
only few of those works have been carefully implemented, and mostly for the $k$-center problem \cite{ChanGS18, ChanGHS22}.
For dynamic \emph{graph} algorithms there is, however, a large body of empirical work, see \cite{DBLP:journals/jea/HanauerHS22} for a survey.

\paragraph{On Coresets}
As mentioned above, the coreset paradigm is an efficient way to solve clustering in big data, as they allow to ``turn big data into tiny data" \cite{FeldmanSS20}, on which the clustering problem can be efficiently solved. Introduced by \cite{HaM04}, the research on \kmeans~coresets is booming (see e.g. \cite{FeldmanL11, HuangV20, Cohen-AddadSS21, Cohen-AddadLSS22} and references therein), which culminated in optimal coresets for $k$-means clustering of size $O(\eps^{-2} k \min(\eps^{-2}, \sqrt k))$ \cite{Cohen-AddadLSSS22, huang23}, that can be computed as fast as any constant-factor approximation algorithm. Using the \kmpp~algorithm, the running time is therefore $O(nkd)$. This progress is spreading to related problems as well, such as capacitated clustering \cite{BravermanCJKST022}, clustering of lines \cite{MaromF19} or clustering with missing values \cite{BravermanJKW21}. A more comprehensive introduction to the coreset literature can be found in \cite{feldmanSurvey}.

\subsection{Organisation of the Paper}

We first introduce the relevant definitions and standard algorithm. In \Cref{sec:dynAlgo}, we present the algorithm from \cite{henzinger2020} and the optimization we introduced. We present the result of our experiments in \Cref{sec:exp}.

\section{Preliminaries}
\subsection{Definitions}

We call a weighted set a set  $X \subset  \R^d$ with an associated weight function $\weight: X \rightarrow \R^+$.

A \emph{solution} to \kmeans~is any $k$-tuple of points in $\R^d$.
The (Euclidean) \kmeans~problem is defined as follows. Given a set of points $X$ with weights $\weight$ and an integer $k$, the goal is to find the solution $S^* \subset \R^d$ that minimizes $\cost_\weight(S, X) := \sum_{x \in X} \weight(x) \min_{s \in S} \dist^2(x, s)$. 
A solution $S$ is an $\alpha$-approximation if it satisfies $\cost_\weight(S, X) \leq \alpha \cdot \cost_\weight(S^*, X)$.
An $(\alpha, \beta)$-bicriteria approximation is a set of $\beta  k$ many centers, with cost at most $\alpha \cost_\weight(S^*, X)$.

\begin{definition}
For $\eps \geq 0$, an $\eps$-coreset for \kmeans~is a set $C \subseteq U$ with weights $\weight'$ such that, for any candidate solution $S$,
\[\cost_{\weight'}(S, C) = (1\pm \eps)\cost_\weight(S, X).\]    
\end{definition}

In the dynamic model, the input changes and the algorithm has to efficiently maintain a solution to \kmeans. 
More precisely, for any set $X \subset \R^d$ with an associated weight function $\weight : X \rightarrow \R^+$, the algorithm should support operations $Insert(x)$ that inserts point $x$ to the current dataset.\footnote{Our algorithm is actually more general and supports insertion of weighted points.}  Additionally, $Delete(x)$ removes $x$ from $X$. 
After each operation, the algorithm should return an $\eps$-coreset. We note that one can then run any (static) approximation algorithm to extract a good solution from the coreset.

\subsection{Static Coresets}\label{sec:static}
Schwiegelshohn and Sheikh-Omar \cite{schwiegelshohn2022} present a comprehensive description of coreset algorithms, with an experimental evaluation. 
As mentioned previously, $\eps$-coresets with optimal size  $O(\eps^{-2} k \min(\eps^{2}, k^{1/2}))$ can be constructed in time $\tilde{O}(|X|kd)$ \cite{Cohen-AddadLSSS22, huang23}.

The result of \cite{schwiegelshohn2022} is that the \emph{Sensitivity Sampling} algorithm from Feldman and Langberg \cite{Feldman2011} performs best. This algorithm has a theoretical running time of $\tilde{O}(|X|kd)$, and produces an $\eps$-coreset of size essentially $\frac{kd}{\eps^4}$ many points \cite{Feldman2011}.
One of the strength of this algorithm is its simplicity: to compute an $\eps$-coreset, it first computes an $(O(1), O(1))$-bicriteria approximation, and then samples points (essentially) proportionate to their cost in the bicriteria solution.  
The initial bicriteria approximation can be computed using the \kmpp~algorithm. This algorithm works in rounds, each round sampling one point proportionate to the squared distance to the points already sampled. This is guaranteed to be a $O(1)$-approximation after $2k$ iterations \cite{wei2016constant}. Following this, one iteration of \emph{Lloyd's Algorithm} is performed: each center is replaced by the optimal center for its cluster, which is the center of mass of the cluster. While this does not  formally improve the quality of approximation, is results in a significantly better solution in practice. 
Note that, \cite{Feldman2011} showed that, when (roughly) $\frac{kd}{\eps^4}$ points are sampled using sensitivity sampling, they form an $\eps$-coreset (with each point weighted by the inverse of its sampling probability). We give the pseudo-code of the algorithms in \Cref{ap:code}

\section{The Dynamic Algorithms}\label{sec:dynAlgo}

\subsection{Description of the base algorithm.} \label{par:dyn_decription}
A thorough description of the dynamic framework can be found in the original paper \cite{henzinger2020}. We provide here a detailed overview (without proofs) in order to explain our heuristic improvements.
The algorithm relies on two fundamental properties of coresets, that allow to compose coresets together:

\begin{property} \label{lemma:base}
   Let $C_1$ be an $\eps$-coreset for $X_1$, and  $C_2$ be an $\eps$-coresets for $X_2$. Then, $C_1 \cup C_2$ is an $\eps$-coreset of $X_1 \cup X_2$. 
\end{property}

\begin{property} \label{lemma:base2}
   Let $C_1$ be an $\eps$-coreset for $X$, and  $C_2$ be a $\delta$-coresets for $C_1$. 
   Then, $C_2$ is an $(\eps + \delta + \eps \delta)$-coreset of $X$.
\end{property}

The high-level idea of the algorithm is to use  a combination of those two properties. To compute an $\eps$-coreset of a set $X$, one can split $X$ into two halves $X_1, X_2$, and compute $\eps/2$-coresets $C_1, C_2$ for those. By Property~\ref{lemma:base}, the union of $C_1$ and $C_2$ is an $\eps/2$-coreset for $X$. To reduce the size further, let $C$ be an $\eps/2$-coreset of $C_1 \cup C_2$: Property~\ref{lemma:base2} ensures that $C$ is an $\eps$-coreset of $X$. 
The reason for this ``two-level'' scheme  is that $C_1 \cup C_2$ has size much smaller than $X$, and, thus, maintaining a coreset of it will be  more efficient than maintaining it for $X$.
More specifically, when the dataset is updated by an insert or remove operation, half of the work is already done: if the update happens to be in $X_1$, then $C_2$ is still an $\eps/2$-coreset for $X_2$ and does not need to be recomputed. Leveraging this idea recursively allows to dynamically maintain a coreset with essentially $O(\log n)$ applications of Property~\ref{lemma:base2}.

More formally, the dataset is hierarchically decomposed in a bottom-up fashion as follows: the algorithm maintains a partition of the dataset $X$ into parts of size $1$, that are recursively assembled in a binary tree structure. A tree node $v$  \emph{represents} the set of points contained in the leaves of the subtree rooted at $v$, and  we denote by $p(v)$  the parent of $v$, and by $\gamma_1(v)$, $\gamma_2(v)$ the children of $v$ in the tree. 

The algorithm ensures that, after processing an update, the following invariant holds: for each tree node $v$ at height $\ell$, there is a set $C(v)$ that is a $\frac{\eps \ell}{\lceil \log n\rceil}$-coreset of the points represented by $v$. 
This ensures the correctness of the algorithm: for the root $r$, the set $C(r)$ is an $\eps$-coreset of the current dataset.
This invariant is maintained using the following operations for point insertions and deletions.
At first we assume that the number of points in the dataset is about $n$.

\medskip

\paragraph{Insert($x$)}: to insert the point $x$, the algorithm adds $x$ to the first empty leaf $v_0$ of the tree.  To restore the invariant the algorithm needs to update every set $C(v)$, for all ancestor $v$ of $v_0$ in the tree. For this, let $v_1, ..., v_{\log n}$ such that $v_i = p(v_{i-1})$ be the sequence of ancestors of $v_0$ ordered by height. The sets $C(v)$ are updated in this order: first, let $C(v_0)$ be an $\eps/\lceil \log n\rceil$-coreset for $v_0$. Then, for $i = 1$ to $\log n$, let $C_1$ and $C_2$ be the $\frac{\eps (i-1)}{\lceil \log n\rceil}$-coresets of $\gamma_1(v_i)$ and $\gamma_2(v_i)$: let $C(v_i)$ be an $\eps/\lceil \log n\rceil$-coreset of $C_1 \cup C_2$. Using  Properties~\ref{lemma:base} and \ref{lemma:base2}, this is indeed an $\frac{\eps i}{\lceil \log n\rceil}$-coreset of the points represented by $v_i$, and therefore the invariant is satisfied at the end of this procedure.

\paragraph{Delete($x$)}: the algorithm identifies (with e.g. a hashtable) in which leaf $x$ is stored, and removes it. Afterwards, all coresets of parents of the leaf involved are updated similarly as for an insertion, in order to maintain the invariant.

Details of how to delete and insert leafs are given in \Cref{par:ap_illustrations}.

\medskip

In case the number of points in the dataset is not roughly $n$ but varies a lot, the work of the dynamic algorithm is partitioned into \textit{phases}. Each phase ends after the number of points in the dataset changed by $50\%$ compared to the beginning of the phase. At the beginning of a phase, each set $C(v)$ is recomputed in order to maintain the invariant regarding the core-set size -- as the value of $\lceil \log n\rceil $ changes between phases, the precision of the coreset in each tree node also has to change.
The recomputation time can be amortized over the updates that caused the change in $\lceil \log n\rceil$. By performing the recomputation ``spread out'' over subsequence updates, the time complexity can even be turned into a worst-case bound (see   \cite{henzinger2020}), but as this is not the focus of our work, we did not implement this.

\medskip
This algorithm uses a static coreset construction as a blackbox. As explained in \Cref{sec:static}, the state-of-the-art is the algorithm from \cite{Feldman2011}, with a bicriteria solution computed with \kmpp. We use this algorithm both for the dynamic algorithm and as a baseline. We note that better results may be achievable with a different coreset construction: however, we expect the relative comparison between the different dynamic algorithms to stay alike.

\subsection{Running time optimizations.} \label{par:dyn_opt}

We propose several improvements to the implementation of the previous algorithm, some are quite standard and some are new heuristics. Our first optimization is to compress the lower levels of the tree: instead of each leaf representing a single point, we introduce a parameter $s$, and make each leaf represent between $s/2$ and $s$ points -- except for one special leaf that may contain fewer than $s/2$ points. All insertions are performed in the special leaf, and whenever it reaches size $s$, it is turned into a ``normal" leaf, and a new special one is created.
This requires to slightly change the insert and delete procedures.
Furthermore, instead of computing a precise $\eps/\log n$-coreset, we parameterize our experiments directly by the size of the coreset $s$. Taking $s \approx \frac{kd \log n}{\eps^2}$ ensures that we compute an $\eps$-coreset with some good probability. In practice much smaller coreset size work just as well, as our experiments show.

Our main approach for optimization is to perform lazy updates in order to group the recomputations together: we show how to maintain the coreset while avoiding some recomputations. Our improvement for insertion provably maintains a coreset, while the deletion case is only heuristic.

\subsection{Optimizing Insertions.} \label{par:opt_ins}
For insertions, each tree node has an \emph{epoch}: an epoch starts (i) either after $s$ points have been inserted into the set represented by the node, (ii) when a point is removed from the node or (iii) when the coreset of the node is recomputed for another reason (e.g. beginning of a new phase). If a new epoch is started for one node, the same is done for all ancestors of this node. 
During an epoch, the tree node maintains two coresets: one for the points present at the beginning of the epoch, and one for newly inserted points. Since there are at most $s$ inserted points, it is very easy to maintain the latter: it merely consists of all inserted points, with their initial weights. The former is constant, as a new epoch is started if a point is removed. 
Property~\ref{lemma:base} ensures that the union of those two sets is a coreset for the whole set represented by the tree node. Therefore, the algorithm stays correct.
Furthermore, at the start of a new epoch, a coreset of size $s$ is recomputed for all points that are currently represented by the node; and the coreset for newly inserted points is set to $\emptyset$.

Although this optimization has no effect in the worst-case scenario, it significantly improves the running time in insertion-only streams, as shown in the next lemma. This hints that this optimization is beneficial when several points are added in a row.

\begin{lemma}
    In a stream of $n$ insertion, the algorithm described above has total running time $\tilde O(nkd)$, and maintains a valid $\eps$-coreset at each time step.
\end{lemma}
\begin{proof}
    Running the static coreset algorithm on a set of size $O(s)$, to produce a coreset with size $s$, has running time $O(skd)$. 
    To recompute the coreset of a node at the beginning of an epoch, the algorithm collects all coresets maintained by the children of this node -- there are $4$ of them, $2$ for each children -- and then  computes a coreset of size $s$ from the union of those $4$ coresets. This takes time $O(skd)$, as explained above. 
    Furthermore, any tree node has at most $O(\log n)$ many ancestors in a tree. Therefore, if a new epoch is started at a node, $O(\log n)$ coreset are recomputed, and the total complexity is $O(s k d \log n)$. 
    Now, if every inserted point gives a budget $O(kd \log n)$ to each of its $O(\log n)$ ancestors, then, when a new epoch starts from one node, its budget is enough to pay for the total recomputation. Therefore, the overall complexity after $n$ insertions is at most the total budget, which is $\tilde O(nkd)$.

    The set maintained by the algorithm is an $\eps$-coreset: the proof follows from Property \ref{lemma:base} and the fact that, during a given epoch, each node stores a $0$-coreset for the points inserted during that epoch.
\end{proof}

\subsection{Optimizing Deletions} \label{par:opt_del}
We propose the following heuristic. 
When $p$ is deleted, the algorithm first checks if $p$ is part of the \emph{final} coreset. If this is not the case, the point is marked for removal, but no further changes are made. 
Otherwise, the algorithm removes $p$, and all other points currently marked for removal, using the non-optimized algorithm. 
When recomputing the coreset of a node due to an insertion, all points marked and represented by this node are deleted (and unmarked).
We note that, when the final coreset has size much smaller than $n$, many points may be deleted before a recomputation is triggered -- if points are deleted randomly, one needs to delete $O(n/s)$ points before hitting one that is in the final coreset.
However, those deleted points still influence the coreset computation (by being influential at lower levels of the tree). In order to balance the running time improvement and the quality of the coreset, we introduce a cut-off $\Delta$, and proceed to recomputation also when $\Delta \cdot n$ of the total points are marked to be removed. We present in \Cref{par:ex_opt_del} different experiments to optimized $\Delta$ in two particular cases:

$\bullet$ if points are deleted in the same order as they are inserted -- e.g., a sliding windows -- then it is best to set $\Delta \approx s/n $. In theory, this would lead to a speed-up of the order $s$: indeed, $s$ points  inserted consecutively are stored in the same tree leaf by our algorithm. Therefore, instead of recomputing coresets for this leaf and its parents $s$ times, our optimization recompute it only once.

$\bullet$ if points are deleted randomly, then they come from potentially very different leaves and the previous argument does not apply. However, random deletions have a weaker affect on the \kmeans~cost, and it appears possible to take a much larger $\Delta$, e.g., $\Delta = 3\%$. This ensure that, in average, the $k$-means cost changes by only $3\%$, which is the order of error we expect from our algorithm. We show experimentally that such a large $\Delta$ indeed yields to a large speed-up.

Those results hint how to pick $\Delta$ according to the expected evolution of the dataset.

\subsection{Flattening the Tree}\label{par:shallowTheory}
Our algorithm is heavily reliant on a balanced binary-tree datastructure, which has to be recomputed regularly.
Similar tree-like datastructures are commonly used in practice and can often be made more efficient by fixing the number of levels of the tree. Such strategies may increase the work done at each level, but can reduce overhead significantly. 
We tested different ways of flattening the tree into a \emph{shallow tree}. 

We suggest to use a complete $g$-ary tree with a fixed height $h$. To find the optimal degree $g$ of the internal nodes, we balance the work done in the leafs with the work in the remaining nodes as follows: there are $g^h$ many leaves, each of those containing $n/g^h$ many points. Computing a coreset on a leaf costs therefore essentially $k \cdot n/g^h$. On a tree node of arity $g$, the algorithm collects $g$ coresets of size $s$ and computes a coreset of those: therefore, the work done is $k \cdot gs$. To balance this work, we want $g$ such that $k \cdot n/g^h = k \cdot gs$, i.e., $g = \sqrt[h+1]{n/s}$.

This optimization requires a priori knowledge of the size of the dataset $n$: to cope with it, the algorithm can work in phases, and recompute from scratch every time the value of $n$ changes by a large factor. We simplified the setting for our experiments and ensured the number of points stayed close to a fixed value.
Those experiments show that, the shallow tree may provide a speed-up, however highly dependent on the size of the dataset and the parameters $k$ and $s$: typically, for $h=2, k=50, s=90$, the algorithm is faster for $n = 20.000$ but slower for $n=60.000$; this trend is reversed for $h=5$. Precise results are shown in \Cref{ap:shallow}. 
It appears that the improvement is uncertain, as it depends on the size of the dataset. 
Therefore, we find it preferable to have an adaptive height, as in our original algorithm. 
However, this experience shows the potential benefits of an a priory knowledge on the data set size.

\subsection{Space Requirement and Data Structure.} \label{par:memTheory}
In the algorithms as described above, the total memory size is $2nd$: at each node of the tree, $s$ points are stored (which takes memory $sd$), and a binary tree with $n/s$ leaves contains $2n/s$ nodes in total.
However, we use the following approach to reduce  the memory requirement to $nd$, i.e., improve it by a factor 2.
First, instead of storing the $d$-dimensional points in each node, one can store pointers to those input points.\footnote{More precisely, we store each input point in a hashtable with keys the identifier of each point, in order to ensure fast deletions of the points.} This works as long as coreset points are input points. This may not always be the case, as centers from the bicriteria approximation algorithm also are coreset points.
One can modify the used bicriteria algorithm to ensure that all those centers are input points by removing the Lloyd step of the \kmpp~algorithm. The drawback is that the quality of the bicriteria approximation worsen in practice, which in turn may worsen the quality of the coreset.
The weight (which fits into a single variable) still need to be stored explicitly in each tree node. With those improvements,the memory requirement is only $nd + n$ ($nd$ for the original points and $n$ for all the pointers).  We tested this in \Cref{par:ex_lessMemory}.

\section{Experiments}

\subsection{Experimental Setup}

\subsubsection{Performance metrics.}
In all experiments, we measured (1) the distortion of the coreset, (2) the quality of the $k$-means solution computed on the coreset, and (3) the running time of maintaining both the coreset and a $k$-means solution. All experiments were repeated five times and the average result is reported. For the running time measurements, we timed the computation of the coreset itself and the computation of the \kmeans~solution  with the C\texttt{++} library Chrono. (We did not include the time needed to evaluate the distortion and quality.) Note that we use a logarithmic scale whenever we plot running time results.

\paragraph{Coreset distortion:}
To evaluate how well a coreset $\coreset$ represents the original datapoints of dataset $X$, one can consider a set $\calS$ of candidate solutions, and evaluate the \emph{distortion} of the coreset on each of those solutions as follows:
\begin{align*}
D_\coreset = \max_{S \in \calS} \left(\max \left(\frac{\cost(S,X)}{\cost(S,\coreset)},\frac{\cost(S,\coreset)}{\cost(S,X)} \right) \right) -1
\end{align*}

If $\calS$ is the set of all possible solutions, this definition ensures that $C$ is a $D_C$-coreset for $X$.
However, it is impossible to enumerate efficiently over all those solutions -- and \cite{schwiegelshohn2022} showed the co-NP-hardness of checking whether a given set is a coreset. To cope with this, \cite{schwiegelshohn2022} proposed to consider $\calS$ to be a single solution, computed via \kmpp~algorithm on the  coreset $C$. 
We slightly extend $\calS$ by also adding a solution computed by \kmpp~on the full dataset $X$. 
Note that \cite{schwiegelshohn2022} showed that adding solutions generated uniformly at random under some natural distributions to $\calS$ was pointless, as for those random solutions the quality of the coreset is always very good. 

\paragraph{\kmeans~quality:}
As our coreset can be used to dynamically compute a solution $S_\coreset$ to \kmeans, we  compare the cost of $S_\coreset$ to a solution $S_X$ computed on the full dataset as follows:
$$Q_\kmeans = \cost(S_X,X) / \cost(S_\coreset,X).$$ 
To get the solution $S_X$, we use again the \kmpp~algorithm.

\subsubsection{Data sets and Update Sequences.}\label{par:datasets}
Since most data sets available for clustering are not intrinsically dynamic, we followed different (standard) strategies to simulate a dynamical behaviour, starting from an ordered data set, 

\begin{enumerate}[noitemsep]
\item An \emph{insertions-only} data set can be generated by inserting each point after another in the same order as they occur in the static dataset.

\item A \emph{sliding window} of size $t$ can be used to include removal operations. First $t$ points are inserted without any point removals. In the next phase (called the \emph{window}), operations alternate between removing the oldest point in the dynamic dataset and inserting the next points of the static dataset. 

\item A \emph{random window} with insertion probability $\pi$ is created by inserting the next point of the static dataset with a probability of $\pi$ or removing a random existing point with a probability of $1-\pi$.

\item A \emph{snake window} of size $t$ is created by chaining multiple random windows together. First, a random window with insertion probability $0.9$ is constructed, until the dynamic dataset contains $t$ points. Then, the insertion probability is reduced to $\pi = 0.1$ and operations are added, until $0.2 \cdot t$ points remain in the dynamic data set. This pattern can be repeated, until all points of the static data set are inserted.
\end{enumerate}

These different datasets allow us to highlight various aspects of the algorithm, namely the impact of insertion (for insertion-only), the impact of making deletions in order (for sliding windows) versus random deletions (random window), and the impact of varying the size of the dataset (snake window).

\emph{Synthetic data set.} To optimize the parameters of the algorithms and show details of the running time behaviour, we used a subset of a synthetic data set with two dimensions, called \emph{Birch} \cite{Birchsets}, which consists of $100.000$ points having $100$ random Gaussian clusters with random size in $\R^2$. For this dataset, we order the points as follows: All points in the same cluster are placed consecutively in the order. Then we created and tested both an insertions-only as well as a sliding window dataset of size $20000$ with a $k$ of $10$. Due to the ordering and the fact that the number of ground truth clusters in the data set is larger than $k$, the optimal $k$-means solution changes over time. 
Finally, to test the algorithm in a more challenging setting, we also tested  a snake window of size $20.000$ with $80.000$ operations applied in total. We call the latter \emph{Birch-snake}.

\emph{Real-world data sets.} With the parameters chosen on the synthetic data sets we then tested our algorithms on larger real-world data sets from the UCI Database \cite{UCI_Data}. We used the datasets (1) \emph{Taxi} \cite{taxi}, which stores the start coordinates of taxi rides in two dimensions, (2) \emph{Twitter} \cite{twitter1,twitter2}, which contains location information of tweets in two dimensions, and (3)  \emph{Census} \cite{census} dataset, which has 68 dimensions. For all these data sets we removed duplicate points and reordered the points randomly. 

We constructed dynamic data sets for Twitter and Census by introducing a snake window of size $t$, where $t \in \{0.5\text{e}6, 1\text{e}6, 1.5\text{e}6\}$. To ensure that the size of the dataset stays roughly 
constant (and be able to have an indication of the running time for a given data set size),  we modified  the description above and kept the insertion probability at $0.1$ until $0.95 \cdot t$ points remain in the dataset.  To avoid side effects due to the construction of the dataset, we start the measurements only after the first $t$ operations. We then measured the performance of the next  $500.000$ operations.

For the data set Taxi, we created a dynamic dataset using a standard sliding windows of size $t \in \{0.5\text{e}6, 1\text{e}6\}$. Measurements were again taken during $500.000$ operations after the first $t$ points had been inserted.

See Table~\ref{tab:params} in Appendix for a listing of all the update sequences we created. The number of clusters $k$ was set to 10. 
Recall that theory recomments to choose $s$ linear in $k \log n$. Thus, we use $s = 5k$ for 
the small synthetic dataset and $s=50k$ for the real-world data sets.

\paragraph{Baseline algorithms.} \label{par:baselines}
We are not aware of any other implementations of dynamic coreset algorithms and, thus, we compare our algorithm with four simple baselines:

$\bullet$ \textit{static}, that computes a coreset from scratch after each update using Sensitivity Sampling in $\calO(|X|kd)$

$\bullet$ \textit{random}, that chooses after each update a uniform sample of the dataset in $\calO(s)$

$\bullet$ \textit{only  k-means}, which computes a \kmeans~solution directly on the data in $\calO(|X|kd)$. Obviously, this does not allow to compare corset distortion. The running time cannot be directly compared to other algorithms, since no coreset is created by this algorithm. However, if the ultimate goal is to compute a $k$-means solution, this baseline  still gives a worthwile comparison.

$\bullet$ \textit{ILP}, which calculates a \kmeans~solution (centers are required to be part of the dataset) directly on the whole data set. We used  the commercial solver Gurobi \cite{gurobi}. This algorithm was at least 10.000 times slower, with a cost close to the one computed by our algorithm. We defer the related discussion to Appendix \ref{par:ex_ILP}

Every experiment was performed on an Intel(R) Core(TM) i5-1235U CPU (4.4 GHz) with 16 GiB of RAM running on Ubuntu 22.04.2 LTS with Kernel 5.19.0-45-generic. All implementations are in C\texttt{++}, compiled using g\texttt{++} version 11.3.0 with the optimization flag -O3. 

\subsection{Analysis of the algorithms and calibration of parameters on small dataset}\label{sec:exp}

In this section, we analyze and compare the implemented algorithms using the small dataset \emph{Birch}. 

Note that the running times of the coreset algorithms include the time to compute a \kmeans~solution on top of the coreset, while the running time of \textit{only  k-means} only measures the time for computing a \kmeans~solution.

\begin{figure*}[t!]
    \begin{subfigure}[b]{0.49\textwidth}
    \centering
    \includegraphics[width=\textwidth]{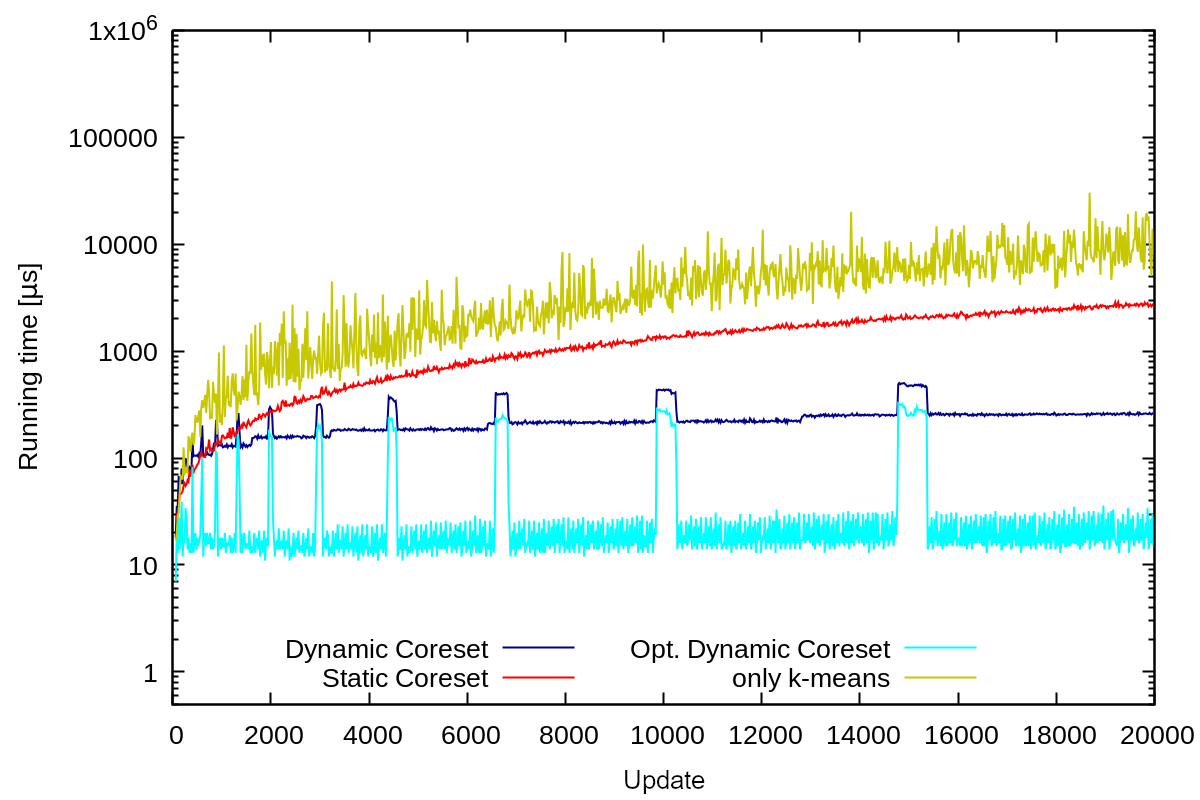}
    \caption{Insertion only dataset with a total of $20000$ operations}
   \label{fig:ex_timeOnlyInsert}
               \begin{minipage}{.1cm}
            \vfill
            \end{minipage}
    \end{subfigure}
    \hfill
    \begin{subfigure}[b]{0.49\textwidth}
    \centering
    \includegraphics[width=\textwidth]{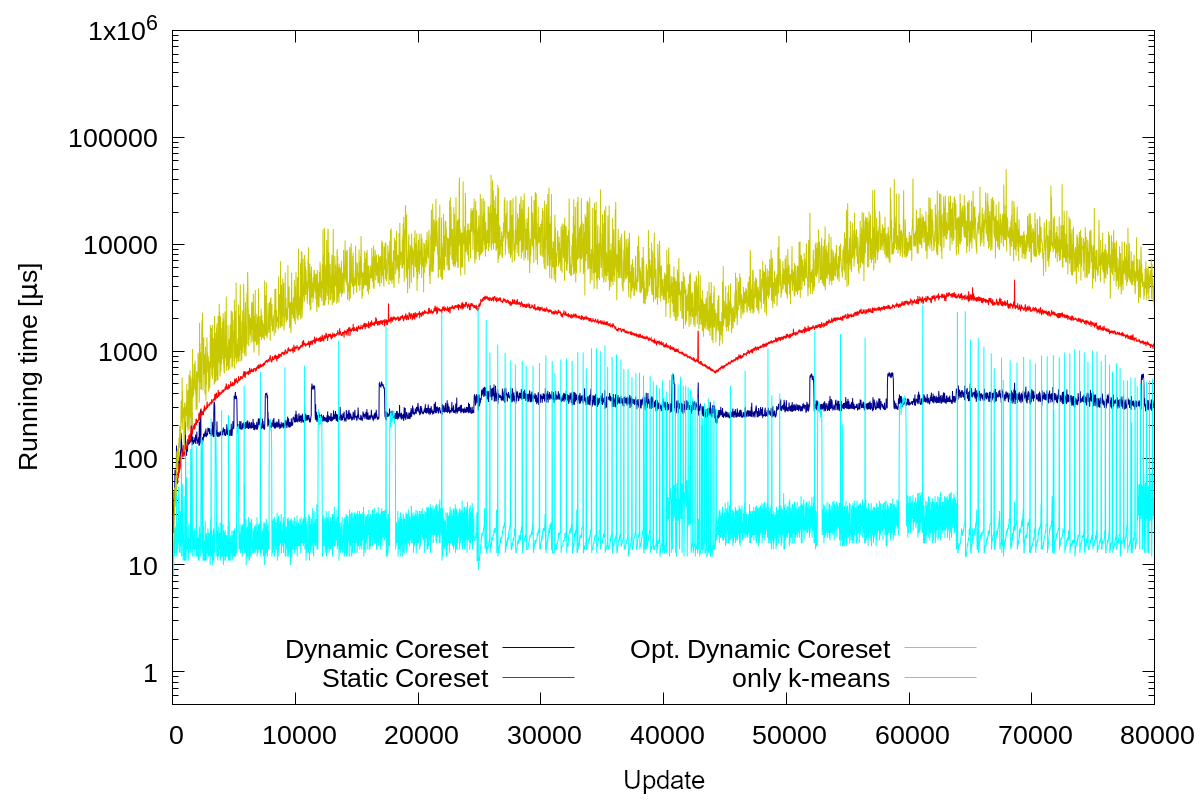}
    \caption{Snake Window with $80000$ operations and size $20000$}    \label{fig:ex_timeSnake}
                \begin{minipage}{.1cm}
            \vfill
            \end{minipage}
    \end{subfigure}
    \caption{\textbf{Running time comparison of all coreset algorithms.} The static data set Birch with $k = 10$ and $s = 50$  was used. The dynamic data sets were created as insert only and snake window as described in \Cref{par:datasets}. The plot shows the running time of the algorithm after each update to the dataset. The running times are averaged over 20 consecutive operations.}
    \label{fig:ex_time}
\end{figure*} 

\subsubsection{Optimizing Deletions} \label{par:ex_opt_del}
We first determine the parameter $\Delta$ needed for the optimized deletion algorithm from \Cref{par:opt_del}.
\Cref{fig:ex_d} presents the running time of the optimized dynamic algorithm for different values of $\Delta$.
For sliding windows, i.e. when points are removed in the same order as inserted, the best running time is achieved for $\Delta = 2s/n$. 
The reason is explained in \Cref{par:opt_del}: points inserted consecutively are stored in the same leaf, leading to a huge reduction in recomputation time, as recomputation happens only once $\Delta \cdot n$ points have been deleted, instead of after each deletion.
\footnote{Taking $2s/n$ instead of $s/n$ allows to save a bit more, as the two leaves concerned by the deletions are consecutive in the tree, and therefore share lots of ancestor, for which only one recomputation is needed with $\Delta= 2s/n$, while $2$ are required for $\Delta = s/n$.}
As  expected, the cost of the \kmeans~solution rises linearly with  $\Delta$, increasing by a factor $1.02$ for $\Delta = 2s/n$ compared to $\Delta = 0$. The plots are presented in \Cref{fig:ex_d} of the appendix.

Using a dataset where points are removed randomly, a much larger $\Delta$ is required to obtain the best possible speedup. However, the cost is also less affected, varying by less than $2\%$ for all tested $\Delta$. The reason is that removing points uniformly at random from all inserted clusters allows a \kmeans~solution to stay accurate for more operations. Exact results are show in the appendix in \Cref{fig:ex_d_random}.
Thus, as the impact on the \kmeans~cost is small, we choose $\Delta = 2s/n$ in the subsequent experiments where deletions appear in order, and $\Delta = 0.03$ otherwise.

\subsubsection{Running time comparison for only insertions.}\label{par:ex_onlyInserts}
The behavior of the algorithms when only inserting points can be seen in \Cref{fig:ex_timeOnlyInsert}.
The experiments confirm the theoretical running times of the different algorithms as described in \Cref{par:baselines}.
Both in theory and in the experiments, the static coreset and calculating the \kmeans~solution directly has a running time that is linear with the input size, while the dynamic algorithms have a logarithmic dependency.

\subsubsection{Running time comparison for sliding window.}\label{par:ex_window}
We evaluated the behavior of the algorithms on a sliding window. The running time of all tested algorithms stayed constant in the window, except for the optimized dynamic which pays every $2s = 100$ deletions for recomputing, and achieve an average speedup of $20$ compared to the base dynamic coreset. The static coreset and only \kmeans~were slower that the optimized dynamic coreset by factors $164$ and $790$ respectively. The exact results are shown in the appendix (\Cref{fig:ex_timeSlidingWindow}).
We note, that the data points are removed in the same order as they were inserted, resulting in optimal conditions for the optimized dynamic algorithm.

\begin{figure*}[t!]
    \centering
    \begin{subfigure}[b]{0.49\textwidth}
    \includegraphics[width=\textwidth]{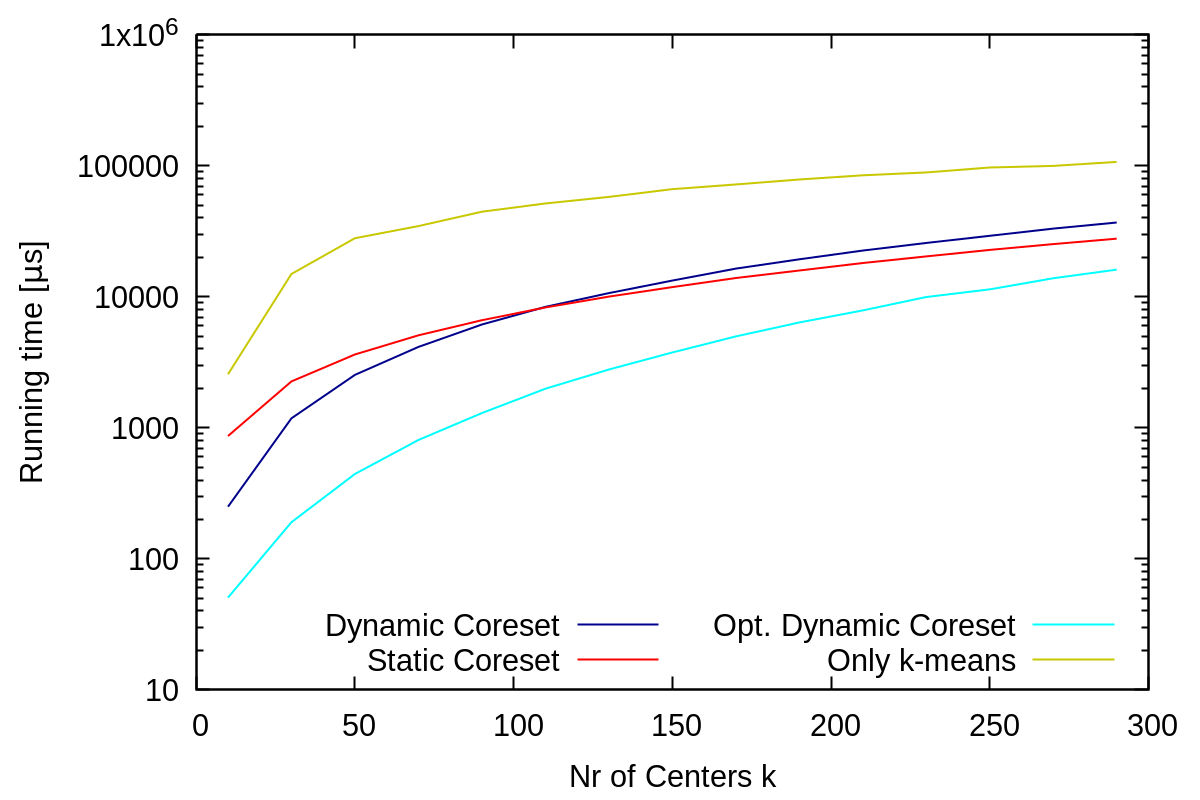}
    \caption{Running time comparison}
    \end{subfigure}
    \hfill
    \begin{subfigure}[b]{0.49\textwidth}
    \centering
    \includegraphics[width=\textwidth]{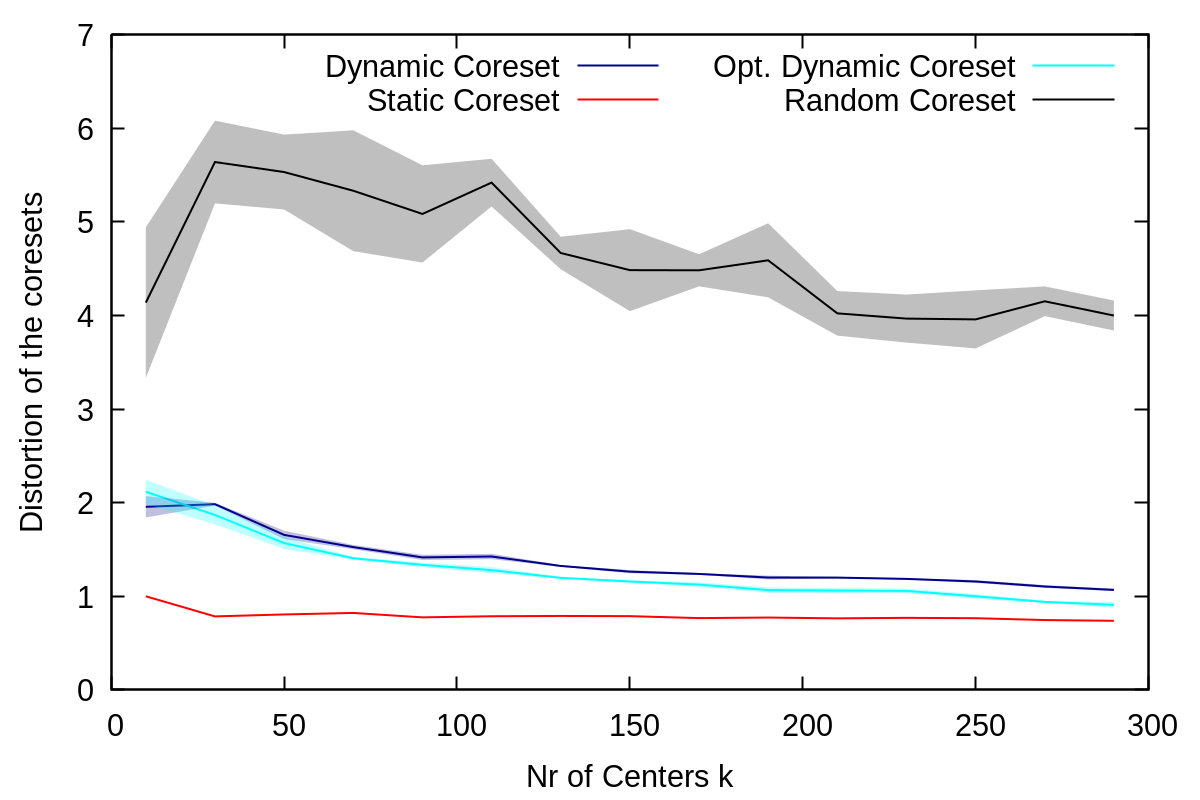}
    \caption{Distortion of the coresets}
    \end{subfigure}
    \caption{\textbf{Impact of \emph{k}.} The Birch-snake dataset was used with a varying \emph{k}. The coreset size was set to $5k$ and time measurements are averages over all update operations after the first $20.000$ points have been inserted. The experiments were repeated five times and the median was used.}
    \label{fig:ex_k}
\end{figure*}

\subsubsection{Running time comparison for Snake Window.} \label{par:ex_snake}
 The snake-windows dataset  represents a setting where both the size and points in dynamic data change drastically. The running time of the different algorithms can be seen in \Cref{fig:ex_timeSnake}. As the  number of points in the dynamic data set periodically decrease and increase in $4000$ operation intervals, both the running time of the static coreset and \kmeans~on the total data set exhibit both a periodic pattern. This is not the case for the dynamic coreset algorithms. The main reason for the lack of symmetry is that randomly removing and inserting points leads to an imbalance in the tree, where each leaf no longer stores the exact same number of points. This slightly reduces the efficiency of both dynamic coreset algorithms. The spikes in the optimized dynamic coreset algorithm indicate when recomputations happened.

\subsubsection{Impact of Coreset Size $s$.} \label{par:ex_compareS}

The effect of the coreset size on running time and \kmeans~cost was tested with $k = 10$ and $s \in \{40, ..., 280\}$. 
We present the exact results in the appendix (\Cref{fig:ex_c}).

The first observation is that the increase in running time with $s$ is faster for the dynamic coresets than for the static coreset. 
This is due to the fact that the dependency in $s$ of the static algorithm is $sk$ (to compute one $k$-means solution on the coreset of size $s$), while the dynamic algorithm needs to compute $\log n$ coreset of size $s$ -- hence running time $s k \log n$. 
With the used settings, the base dynamic coreset is already slower than the static coreset at $s = 350$. However, the number of points in the dynamic data set used is relatively small ($20.000$ points). For experiments with larger datasets, see \Cref{par:ex_large}.

Looking at the distortion of the different coresets, selecting points randomly is beneficial when the size of the coreset $s$ grows and starts competing with the other algorithms when $s = 150$, but can be very poor for small values of $s$.

\subsubsection{Impact of number of clusters $k$.} \label{par:ex_compareK}

The effect of $k$ on running time and coreset distortion is plotted in \Cref{fig:ex_k}. 
For each value of $k \in \{10,..., 280\}$, a point corresponds to the average running time (respectively cost) of each algorithm between operation $20.000$ and $80.000$ of the dataset -- to avoid side effects due to the initial insertions. 

We observe that the value of $k$ has slightly positive impact on the quality of all the algorithms, but their relative order remains the same. 
However, for running time, the non-optimized dynamic algorithm performs poorly for large values of $k$: For $k$ roughly 100 or above it is slower than the static algorithm. This is due to the fact that the dynamic implementation is quite intricate, and its running time scales with $sk \approx k^2$: Thus it cannot compete with the simple coreset construction for large values of $k$. However, this effect is mitigated with our optimized implementation, which is faster than the static algorithm for all values of $k$ we experimented with and also has a better coreset quality than the non-optimized dynamic algorithm.

\subsubsection{Reduction of memory overhead.} \label{par:ex_lessMemory}
We also tested a variant of the algorithm without the Lloyd step in the bicriteria algorithm, as described in \Cref{par:memTheory}. We use $k$ of 10 and $s = 50$ with the Birch-snake dataset.
The results are the following: as expected, the memory footprint of the tree was drastically reduced. However, 
to process deletions, the coreset is stored in a hashtable (unordered map), and, thus, the memory reduction was only $8\%$. Furthermore, the running time was reduced by $~11\%$, and the cost of the \kmeans~solution increases by $5\%$. We decided to privilege quality over memory in the remaining, and did not use this variant for our experiments on large data sets.

\begin{table*}[t!]
\centering
\caption{\textbf{Results of experiments with large datasets.} Speedup of the optimized dynamic coreset over each of the other algorithms, and also \kmeans~quality ($Q_k$) and coreset distortion ($D_\coreset$) for each algorithm. ODyn is the optimized dynamic coreset, Dyn the dynamic coreset, Stat the static coreset, Rand the random coreset and KM is only \kmeans. The datasets have size $t$ and are created
with $k = 10$, $s = 50k$.}
\begin{tabular}{ll|rrr|rr|rr}
\toprule
\multicolumn{2}{c|}{Input} & \multicolumn{3}{|c|}{Twitter} & \multicolumn{2}{|c|}{Taxi} & \multicolumn{2}{|c}{Census}  \\
\multicolumn{2}{c|}{t $[e6]$} & 0.5 & 1.0 & 1.5 & 0.5 & 1.0 & 0.5 & 1.5 \\
\midrule
\multirow{4}{*}{Speed-up}&Dyn&11&12&12&20&23&2.5&2.59 \\ 
&Stat&240&447&580&270&550&158&431 \\
&Rand&0.039&0.051&0.049&0.037&0.054&0.02&0.02 \\
&KM&850&1663&2100&1100&2000&2028&5717 \\
\midrule
\multirow{4}{*}{$Q_k$} &ODyn&0.96&0.96&0.95&0.83&0.85&0.96&0.95 \\ 
&Dyn&0.95&0.96&0.95&0.83&0.84&0.95&0.95 \\
&Stat&0.95&0.95&0.94&0.89&0.89&0.97&0.97 \\
&Rand&0.91&0.92&0.92&0.46&0.16&0.96&0.96 \\
\midrule
\multirow{4}{*}{$D_C$} &ODyn&0.62&0.62&0.65&0.91&1.3&0.66&0.71 \\ 
&Dyn&0.62&0.64&0.67&0.91&1.3&0.66&0.75 \\
&Stat&0.11&0.10&0.10&0.37&0.57&0.09&0.09 \\
&Rand&0.12&0.20&0.21&3.8&22&0.07&0.08 \\
\end{tabular}
\label{tab:ex_large}
\end{table*}

\subsubsection{Shallow Tree.} \label{par:ex_shallow}

We implemented shallow trees of different height $h$ as described in \Cref{par:shallowTheory}. To give maximum advantage to the shallow tree algorithm,  we use a random sliding window with insertion probability $\pi = 0.5$ and size 20000, which keeps the data size roughly constant: no restructuring is necessary and the number of points in each leaf is predetermined. The Brich dataset with $k = 10,50 $ and $s = 50, 90$ was used.
To verify that the degree $g = \sqrt[h+1]{n/s}$ predicted by theory is indeed correct, we compared the running time when $h = 1$ with different values for $g$ using the parameters shown in \Cref{tab:params}. In \Cref{fig:ex_shallowTimeG}, we show that the observed minimum lies very close to the theoretical optimum of $g = \sqrt[h+1]{n/s} = 20$. Therefore, we used the theoretical values for $g$ in all further experiments.

When using the optimal $g$, the shallow tree of height one could compete with the base dynamic coreset. The speedup is however dependent on many factors like the height of the tree, the size of the coreset, $k$ and $s$. We show the speedup of shallow trees with  different heights $h$ and optimal $g$ in \Cref{fig:ex_shallowTimeH} in the appendix. As expected, there exists one optimal set of $g$ and $h$ for each combination of $n$, $k$, $s$ and changing either parameter would require a complete restructuring of the tree and a recalculation of all coresets. Especially when the size of the dataset varies over time, frequent restructuring of the tree would be necessary: this would yield essentially our main algorithm. 
Therefore, we 
did not use a shallow tree  for our experiments on large data sets

\subsection{Experiments on large data sets} \label{par:ex_large}
We present an evaluation of our optimized dynamic algorithm on real-world dataset.
We used the datasets described in \Cref{par:datasets}: Twitter and Census with a snake window and $\Delta = 0.03$ (as determined above), and Taxi with a sliding window and $\Delta = 2s/n$. 
We chose $k=10$ and $k=25$ (the latter is presented in Appendix~\ref{ap:largeDataset} and shows results for the Taxi and the Twitter data set) and $s=50k$.

We show in \Cref{tab:ex_large} the speedup of the optimized dynamic algorithm compared to the baselines, together with the respective  coreset distortion and \kmeans~qualities.

For $k=10$ (resp. 25) the speedup of the optimized dynamic coreset improves with larger input sizes and lies between 2.5 and 23 (resp. 10 and 19) compared to the dynamic coreset, between 158 and 580 (resp.~94 and 367)  compared to the static coreset and between 850 and 5717 (resp.~541 and 1725) for only \kmeans. 

Both for $k=10$ and $k=25$  the performance results are similar to the previous data set:
For both dynamic coresets, the standard and the optimized ones, the distortion is worse than for the static coreset algorithm. However, the quality of the solution remains high, where the dynamic coresets only incur a small quality penalty. \emph{Thus, due to the large speedups that it achieves, the optimized dynamic coreset algorithm is a clear improvement over the static one}. The running time improvement over the only-\kmeans~algorithm is even larger, but the quality of the solution is somewhat worse.  

The performance of the random coreset is very dependent on the structure 
of the used dataset. The fact that the solution quality of the random coreset on the Twitter and Census datasets are comparable to the results of only \kmeans~indicates that the clusters in those datasets have roughly the same size and exhibit no special structure. This is, however, not the case for the Taxi dataset, which indicates more structured clusters. \emph{It is, thus, not advisable to use the a random coreset without extensive a priory knowledge of the dataset.}

\section{Conclusion}
We empirically evaluated the theoretically best algorithm for maintaining a coreset for $k$-means in a fully dynamic setting given in \cite{henzinger2020}.
We show that our optimized variant of the algorithm offers a speed-up of 2 to 23 (increasing with the data set size), without loss in the quality, compared to the straightforward implementation of the dynamic algorithm in \cite{henzinger2020}. With our algorithm we can maintain a coreset as well as a $k$-means solution with a quality comparable to executing a \emph{static} algorithm at every time step, while achieving a speedup of about two orders of magnitude.
Note that this was the state-of-the-art prior to our implementation. 

We provide various strategies to optimize even further if some basic statistics of the dataset are known.
First one is to tune the co called \emph{deletion cutoff} $\Delta$ from our optimized dynamic coreset. While a $\Delta$ of $2s/n$ is sufficient for cases where points are deleted roughly in the same order as they are inserted, larger values of up to $0.03$ are beneficial when deletions occur randomly.

Second, a shallow tree can be used when the size of the data set stays roughly constant over time. Together with a priory knowledge of the size of the data set, the optimal arity of the shallow tree can be calculated and fixed over the computations to achieve a significant speedup. Without precise knowledge of these parameters, we recommend using the optimized dynamic algorithm instead. 

Third, we  present a strategy for reducing memory overhead with only a slight penalty in solution quality. This reduction in quality if however additionally compensated by a lower running time.

A natural question following our work is to extend those experiments to $k$-median clustering. Although theory predicts exactly the same behaviour, to our knowledge it has not been confirmed in practice -- the practical evaluation of coresets from \cite{schwiegelshohn2022} has also been performed only with respect to the $k$-means objective.

Another exciting question rises from our experiments: while both dynamic algorithms generate worse coresets than the static algorithm, surprisingly, the quality of the resulting \kmeans~solution is great. The difference with the static algorithm is at most $1\%$ for the dataset census (\Cref{tab:ex_large}), even when the distortion of the dynamic coreset is $0.72$ compared to $0.09$ for the static coreset (as shown in \Cref{tab:ex_large}).
Understanding and explaining this gap seems a promising avenue: coreset construction may be much better ``when it matters".

\FloatBarrier

\section*{Acknowledgements.} 
\erclogowrapped{5\baselineskip} This project has received funding from the European Research Council (ERC) under the European Union's Horizon 2020 research and innovation programme (Grant agreement No. 101019564 ``The Design of Modern Fully Dynamic Data Structures (MoDynStruct)'' and the Austrian Science Fund (FWF) project Z 422-N, project “Static and Dynamic Hierarchical Graph Decompositions”, I 5982-N, and project “Fast Algorithms for a Reactive Network Layer (ReactNet)”, P 33775-N, with additional funding from the netidee SCIENCE Stiftung, 2020–2024. 

D. Sauplic has received funding from the European Union’s Horizon 2020 research and innovation
programme under the Marie Skłodowska-Curie grant agreement No 101034413.

\newpage
\bibliographystyle{plain}
\bibliography{main}

\newpage
\appendix

\onecolumn
\FloatBarrier
\section{Illustrations of the Algorithms} \label{par:ap_illustrations}
This section gives an illustrative overview of the procedures to insert a leaf (\Cref{fig:treeInsert}) and remove a leaf (\Cref{fig:treeRemove}) from the used tree structure.
\begin{figure}[h]
    \centering
    \includegraphics[width=0.8\textwidth]{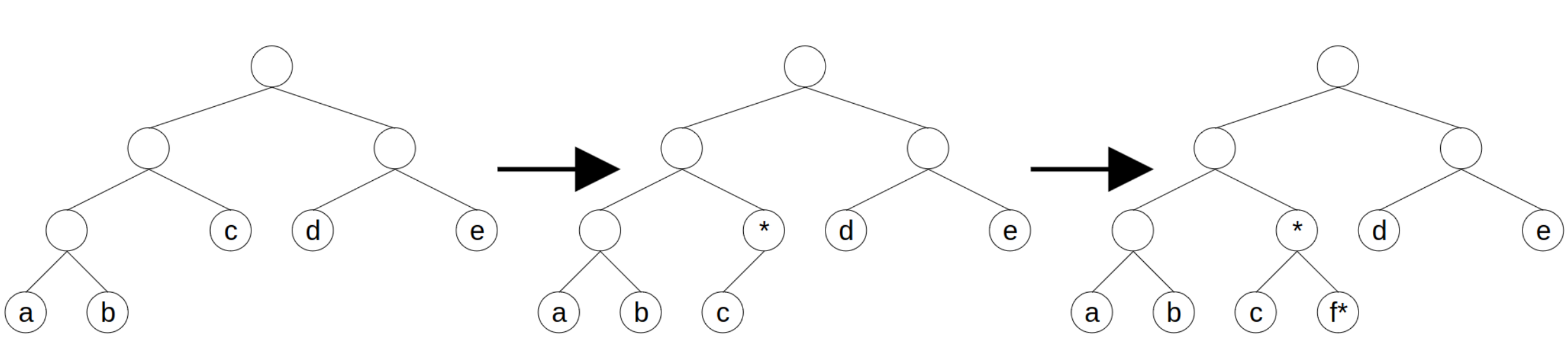}
    \caption{\textbf{Insertion of a new Node.} Nodes are represented as circles and leafs are labeled $a$ to $f$. New nodes are marked with an asterisk. To insert a new leaf, a new child node ($^*$) is added to the leftmost leaf with the lowest depth ($c$) and then swapped with its parent node ($c$). The new leaf ($f^*$) is then added as a sibling of $c$.}
    \label{fig:treeInsert}
\end{figure} 
\begin{figure}[h]
    \centering
    \includegraphics[width=0.8\textwidth]{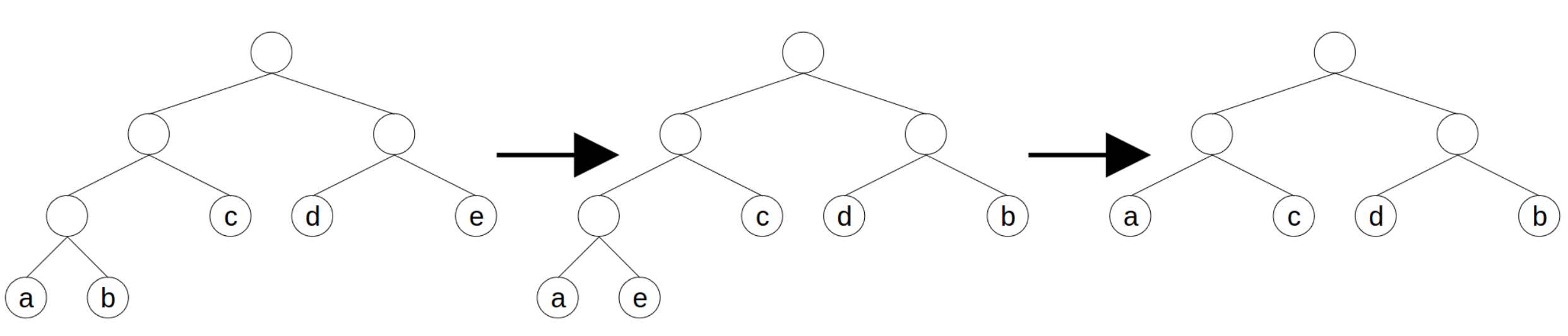}
    \caption{\textbf{Deletion of a Node $e$.} Nodes are represented as circles and leafs are labeled $a$ to $e$. To remove leaf $e$, it is first swapped with the rightmost leaf of the tree with the largest depth ($b$). In the second step, the sibling of $e$ is swapped with its parent node and made a new leaf by deleting all its children. }
    \label{fig:treeRemove}
\end{figure} 
\FloatBarrier

\section{Pseudo-codes} \label{ap:code}

Here, we introduce the algorithms described in the main text as pseudo code. The include the \kmpp~algorithm with input weights (\Cref{alg:KMeans++}) and the algorithm used to select point for the coreset of Sensitivity Sampling (\Cref{alg:Coreset}). We furthermore describe how to insert and remove points from the dynamic coreset and its underlying tree structure (\Cref{alg:Insert}, \Cref{alg:Remove} and helper function \Cref{alg:InnerCoreset}).

\begin{algorithm}[h]
\caption{\kmpp$(X, \weight, n)$}\label{alg:KMeans++}
\begin{algorithmic}[1]
\Require A set $X$ with a weight function $\weight$, and an integer  $k'$
\Ensure A set $S$ with $k'$ centers
\State  Sample a point $x$ randomly from $X$ 
\State Initialize $S \gets x$
\While{$i < k'$}
\State Sample a point $x$ according to the distribution $\weight(x) \cost(x, S) / \cost_\weight(X, S)$
\State Update $S \gets S \cup \{x\}$
\EndWhile
\State Let $S' \gets \emptyset$
\For{Each cluster $C_s$ induced by $s \in S$}
\State $S' \gets S' \cup \{\frac{1}{|C_s}\sum_{p \in C_s} p\}$
\EndFor
\State \Return $S'$
\end{algorithmic}
\end{algorithm}

\begin{algorithm}[h]
\caption{Selecting Coreset Points}\label{alg:Coreset}
\begin{algorithmic}[1]
\Require A set $X$ with a weights function $\weight$, and integer $k$, Number of Coreset Points to select $s$
\Ensure A list of points $\coreset$ and weights $\weight'$, such that the multiset defined by $\coreset$ with weights $\weight'$ is an $\eps$-coreset.
\State $S \gets \kmpp(X, w, 2k)$
\State For any point $x$, let $\weight(S_x)$ be the total weight of its cluster in solution $S$. 
\State Define the sensitivity of $x$ $s(x) = \weight(x) \cost(x, S) / \cost_\weight(X, S) + \frac{1}{\weight(S_x)}$
\For{$i < s$}
\State Sample a point $x$ proportionate to $s(x)$
\State $\coreset[i] \gets x$, $\weight'[i] \gets \frac{\sum_{x'} s(x')}{(s-k) s(x)}$
\EndFor
\For{Each center $s \in S$ with cluster $S_s$:}
\State Add $s$ to $\coreset$ with weight $(1+\eps)\weight(S_s) - \weight'(S_s)$
\EndFor
\State \Return $\coreset, \weight'$
\end{algorithmic}
\end{algorithm}

\begin{algorithm}[h]
\caption{UpdateInnerCoreset($v, k, s$)}\label{alg:InnerCoreset}
\begin{algorithmic}[1]
\Require A tree node $v$, the cardinality of the internal coresets $s$ and the number of centers $k$. We are also given an algorithm Coreset($X,w,k,s$), that, for any set of points $X$ with weights $w$, returns a weighted coreset of size $s$ for $k$-means on $X$.
\If{ $v$ is a leaf}
 \State $C(v), w'(v) \gets $Coreset($C(v), w'(v), k, s$)
\Else
  \State $C_{init} \gets C(c_1(v)) \cup C(c_2(v))$ 
  \State $w \gets w(c_1(v)) \cup w(C_2(v))$
  \If{$|C_{init}| \leq s$} 
    \State $C(v) \gets C_{init}$, $w'(v) \gets w$
  \Else
    \State $C(v), w(v) \gets $Coreset$(C_{init}, w, k, s)$
  \EndIf
\EndIf
\State UpdateInnerCoreset($p(v), k, s$)
\end{algorithmic}
\end{algorithm}

\begin{algorithm}[h]
\caption{Insert(x, w, s, k)}\label{alg:Insert}
\begin{algorithmic}[1]
\Require A point $x$ with weight $w$ to insert into the tree $t$. The cardinality of the coresets $s$ and the number of centers $k$. 
\State We store $i$, which points to a leaf, which contains less than $s$ points or $NULL$ if no such leaf exists.
\If{$i == NULL$}
 \State Let $l$ be the leftmost leaf with the lowest depth in tree $t$
 \State Let $a, b$ be new nodes
 \State $c_1(l) \gets a$, $c_2(l) \gets b$
 \State $p(a) \gets l$, $p(b) \gets l$
 \State $C(a) \gets C(l), w'(a) \gets w'(l)$
 \State $i \gets b$
 \State $UpdateInnerCoreset(l,k,s)$ 
\EndIf
\State $C(i) \gets C(i) \cup x$ 
\State $w'(i) \gets w'(i) \cup w$ 
\State $UpdateInnerCoreset(p(i), k, s)$

\If{$|C(i)| \geq s$}
 \State $i \gets NULL$ 
\EndIf

\end{algorithmic}
\end{algorithm}

\begin{algorithm}[h]
\caption{Remove(x, s, k, v)}\label{alg:Remove}
\begin{algorithmic}[1]
\Require A point $x$ stored in leaf $v$ to remove from the tree $t$. The cardinality of the coresets $s$ and the number of centers $k$. 
\State We store $i$, which points to a leaf, which contains less than $s$ points or $NULL$ if no such leaf exists.
\State $C(v) \gets C(v) \setminus x$ 
\State $UpdateInnerCoreset(p(v), k, s)$
\If{$|C(v)| \leq s/2$}
\hspace{10pt} \State Let $r$ be the rightmost leaf with the highest depth in tree $t$ and $x$ its sibling
\hspace{10pt} \State $C(i) \gets C(i) \cup C(v)$, $w'(i) \gets w'(i) \cup w'(v)$
\hspace{10pt} \State $C(v) \gets C(r)$, $w'(v) \gets w'(r)$
\hspace{10pt} \State $swap(s, p(x))$
\hspace{10pt} \State Make $x$ a leaf
\hspace{10pt} \State $UpdateInnerCoreset(x,k,s)$ 
\hspace{10pt} \State $UpdateInnerCoreset(i,k,s)$ 
\hspace{10pt} \If{$|C(i)| \geq s$}
\hspace{20pt} \State $i \gets NULL$ 
\hspace{10pt} \EndIf

\EndIf
\end{algorithmic}
\end{algorithm}

\FloatBarrier
\section{Experimental Parameters}

\begin{table}[h!]
\centering
\caption{\textbf{Overview of experimental parameters}}
\begin{tabular}{l|r|r|r|r|r|r}
\toprule
Experiment & Dataset & Updates & size & k & s & $\Delta$ \\ 
\midrule
\hyperref[par:ex_opt_del]{Optimizing deletions} & Birch & sliding window & 10000 & 10 & 50 & 0-0.14\\ 
\hyperref[par:ex_onlyInserts]{Running time only inserts} & Birch & only inserts & 20000 & 10 & 50 & -- \\ 
\hyperref[par:ex_window]{Running time sliding window} & Birch & sliding Window & 20000 & 10 & 50 & $2s/n$ \\ 
\hyperref[par:ex_snake]{Running time snake window} & Birch & snake window & 20000 & 10 & 50 & 0.03\\ 
\hyperref[par:ex_ILP]{ILP} & Birch & Only inserts & 300 & 10 & 20 & -- \\ 
\hyperref[par:ex_compareS]{Impact of s} & Birch & snake window & 20000 & 10 & 40-280 & 0.03\\ 
\hyperref[par:ex_compareK]{Impact of k} & Birch & snake window & 20000 & 10-280 & $5k$ & 0.03\\ 
\hyperref[par:ex_shallow]{Shallow Tree impact of $g$} & Birch & sliding window & 20000 & 10 & 50 & -- \\ 
\hyperref[par:ex_shallow]{Shallow Tree impact of $h$} & Birch & sliding window & 20000-80000 & 10,50 & 50,90 & -- \\ 
\hyperref[par:ex_lessMemory]{Memory reduction} & Birch & snake window & 20000 & 10 & 50 & -- \\ 
\hyperref[par:ex_large]{Large Twitter} & Twitter & snake window & 0.5e6-1.5e6 & 10 & 500 & 0.03\\ 
\hyperref[par:ex_large]{Large Taxi} &  Taxi & sliding window & 0.5e6-1.0e6 & 10 & 500 & 2s/n\\ 
\hyperref[par:ex_large]{Large Census} & Census & snake window & 0.5e6-1.5e6 & 10 & 500 & 0.03
\\
\end{tabular}
\label{tab:params}
\end{table}

\FloatBarrier
\section{Additional experimental Data}
This section includes additional experimental data that is not part of the main body.

\FloatBarrier
\subsection{Large experiments.}\label{ap:largeDataset}
Here, we give the experimental results using large datasets and  $k = 25$.

\begin{table}[h!]
\centering
\caption{\textbf{Results of experiments with large datasets.} Speedup of the optimized dynamic coreset over each of the other algorithms, and also \kmeans~quality ($Q_k$) and coreset distortion ($D_\coreset$) for each algorithm.  ODyn stands for optimized dynamic coreset, Dyn for dynamic coreset, Stat for static coreset, Rand for random coreset and KM for only \kmeans. The datasets have size $t$ and are extracted from the static datasets Twitter and Taxi with $k = 25$, $s = 50k$ as described in  \Cref{par:datasets}.}
\begin{tabular}{ll|rrr|rr}
\toprule
\multicolumn{2}{c|}{Input} & \multicolumn{3}{|c|}{Twitter} & \multicolumn{2}{|c|}{Taxi}  \\
\multicolumn{2}{c|}{t $[e6]$} & 0.5 & 1.0 & 1.5 & 0.5 & 1.0 \\
\midrule
\multirow{4}{*}{Speed-up}&Dyn&10&11&12&18&19 \\ 
&Stat&106&229&367&94.0&176 \\
&Rand&0.02&0.04&0.04&0.02&0.03 \\
&KM&541&1090&1725&709&1275 \\
\midrule
\multirow{4}{*}{$Q_k$}&ODyn&0.94&0.94&0.95&0.88&0.81 \\ 
&Dyn&0.95&0.95&0.94&0.88&0.81 \\
&Stat&0.95&0.94&0.93&0.91&0.87 \\
&Rand&0.87&0.89&0.86&0.16&0.05 \\
\midrule
\multirow{4}{*}{$D_\coreset$}&ODyn&0.53&0.53&0.58&0.53&0.85 \\ 
&Dyn&0.54&0.54&0.60&0.57&0.86 \\
&Stat&0.10&0.11&0.11&0.23&0.39 \\
&Rand&0.20&0.18&0.22&12.6&67.1 \\
\end{tabular}
\label{tab:ex_large_k25}
\end{table}

\FloatBarrier
\subsection{Small experiments.}

Here, we show the experiments on small datasets used to optimizes the parameters of the tested algorithms. \Cref{fig:ex_c} contains the impact of the parameter $s$ on the dynamic algorithms. We show the detailed running time of the various algorithms in a sliding widow in \Cref{fig:ex_timeSlidingWindow}.

\begin{figure}
    \centering
    \begin{subfigure}[b]{0.49\textwidth}
    \includegraphics[width=\textwidth]{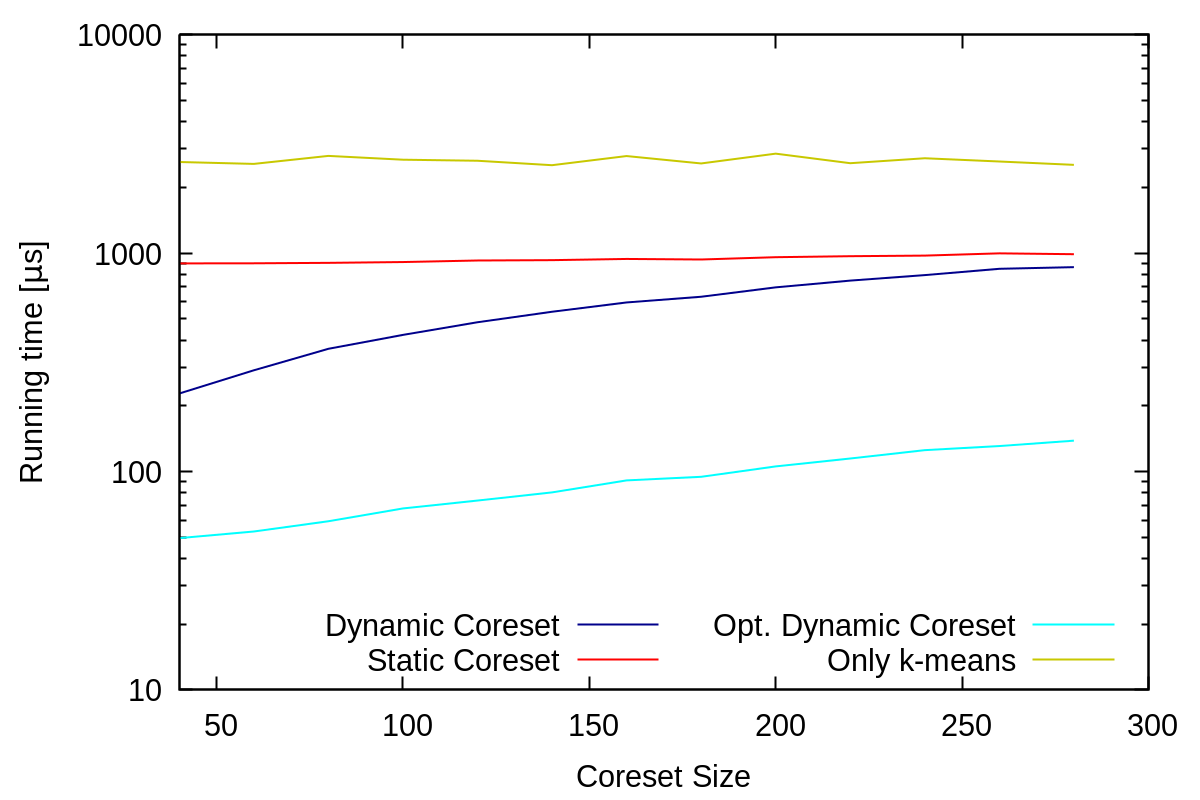}
    \caption{Running time comparison}
    \end{subfigure}
    \hfill
    \begin{subfigure}[b]{0.49\textwidth}
    \centering
    \includegraphics[width=\textwidth]{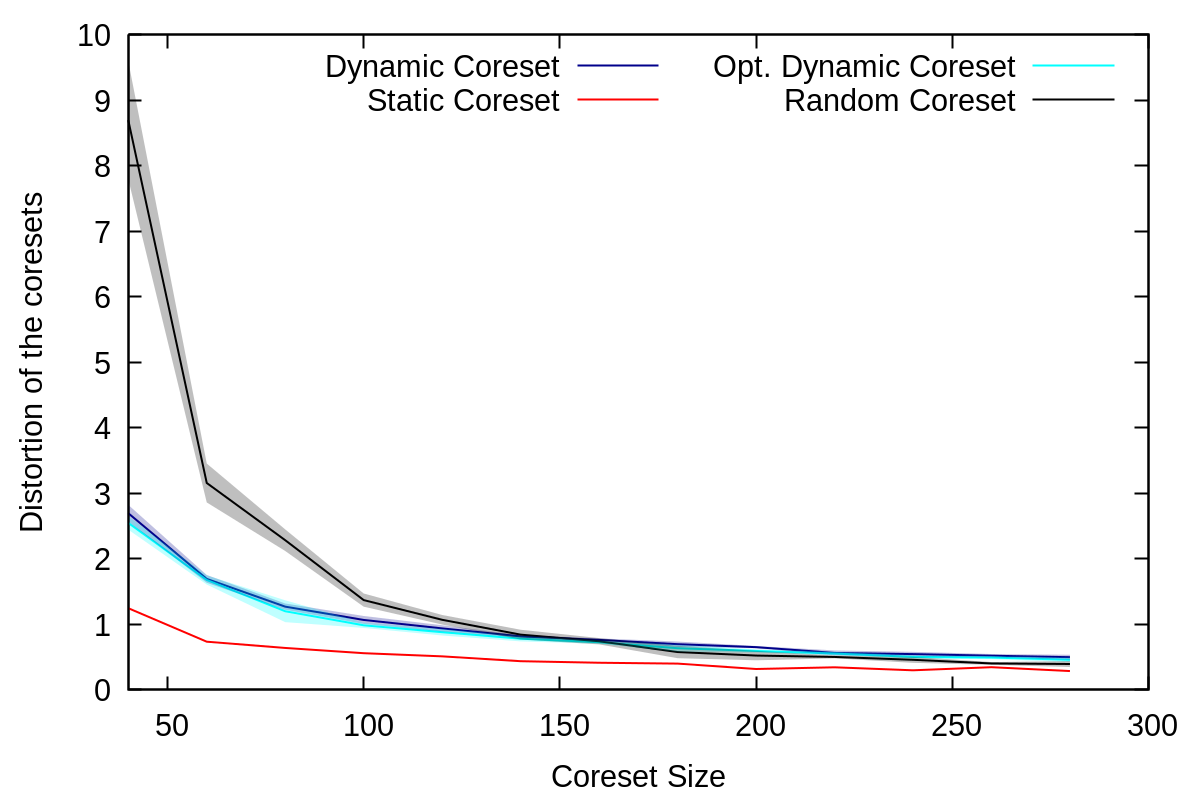}
    \caption{Distortion of the coresets}
    \end{subfigure}
    \caption{\textbf{Impact of \emph{s}.} The Birch-snake dataset with the parameters given in \cref{tab:params} was  used with a varying \emph{s}. Measurements are averaged over all update operations inside of the window. The experiments were repeated five times and the median was used.}
    \label{fig:ex_c}
\end{figure}

\begin{figure}[h]
    \centering
    \includegraphics[width=0.7\textwidth]{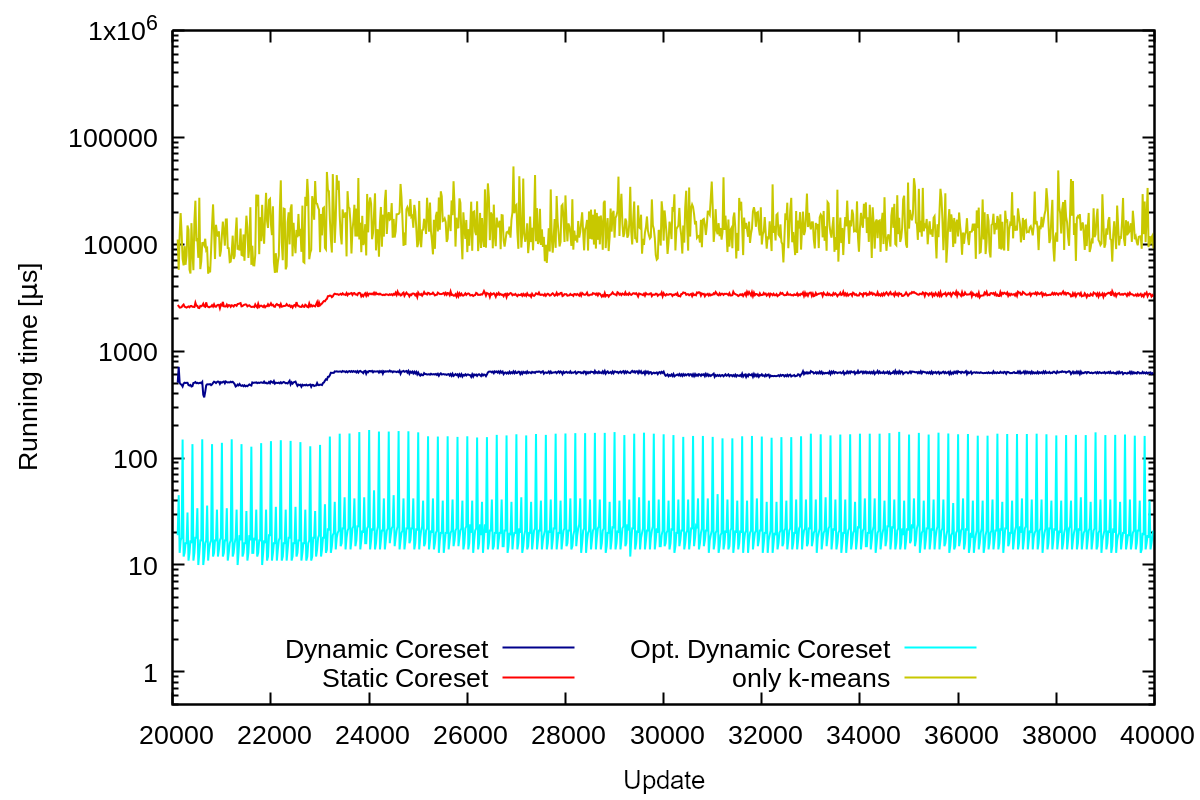}
    \caption{\textbf{Running time comparison of all coreset algorithms using a sliding window.} The static data set Birch with $k = 10$ and $s = 50$ was used. The plot shows the running time of the algorithm after each update to the dataset.}
    \label{fig:ex_timeSlidingWindow}
\end{figure}

\begin{figure}[h]
    \centering
    \begin{subfigure}[b]{0.49\textwidth}
    \includegraphics[width=\textwidth]{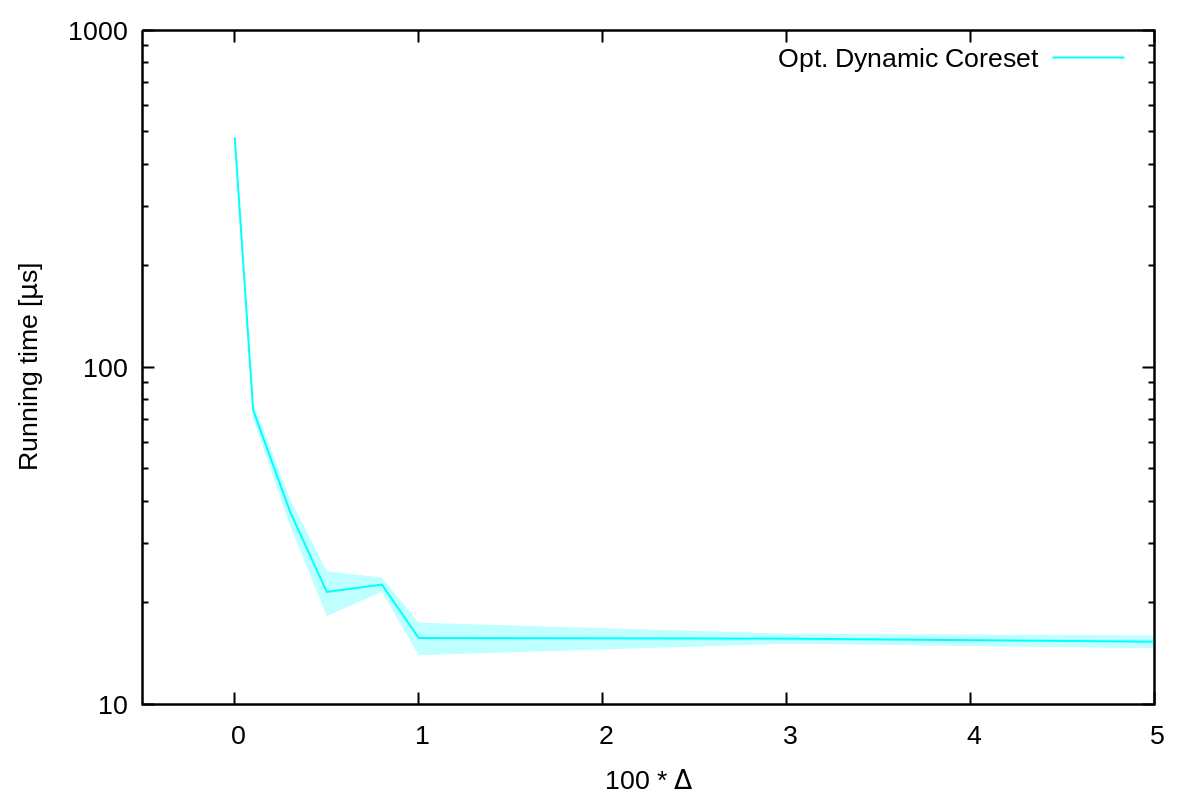}
    \caption{Running time of the sliding Window with ordered deletions}
    \label{fig:ex_d_ordered}
    \end{subfigure}
    \hfill
    \begin{subfigure}[b]{0.49\textwidth}
    \centering
    \includegraphics[width=\textwidth]{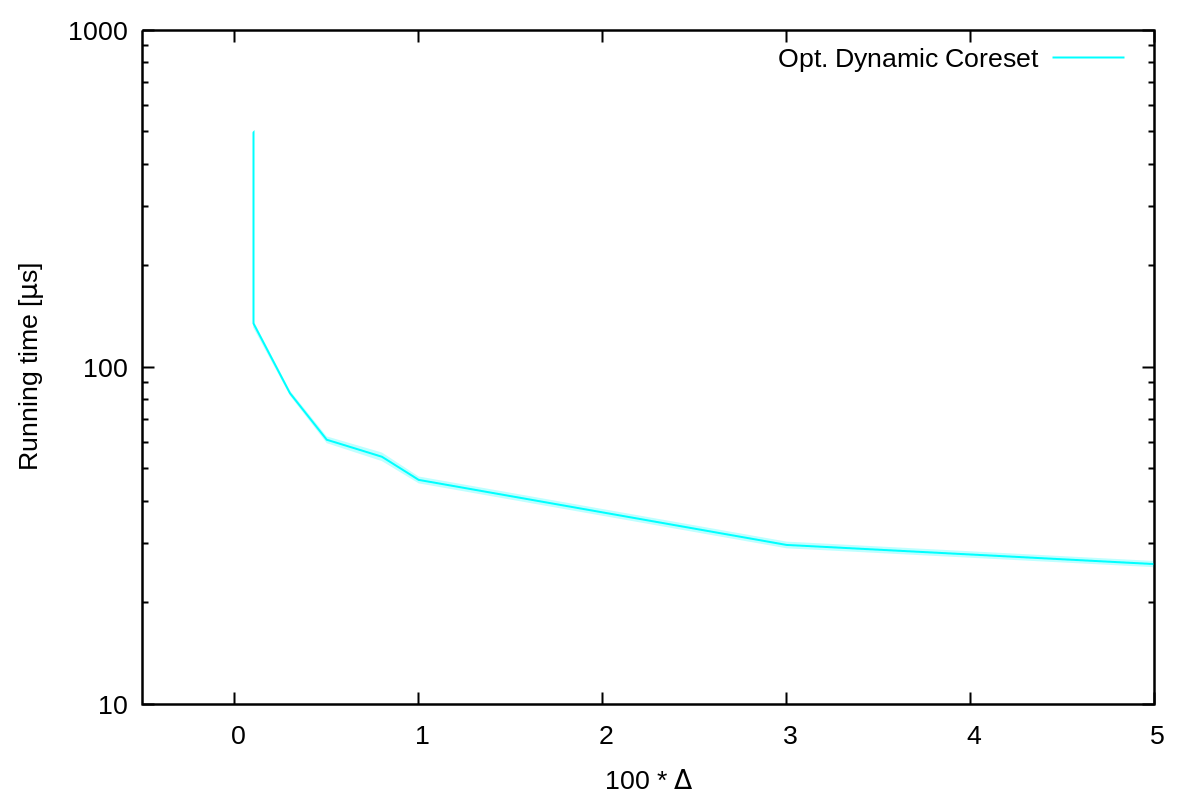}
    \caption{Running time of the sliding Window with random deletions}     \label{fig:ex_d_random}
    \end{subfigure}
    \begin{subfigure}[b]{0.49\textwidth}
    \centering
    \includegraphics[width=\textwidth]{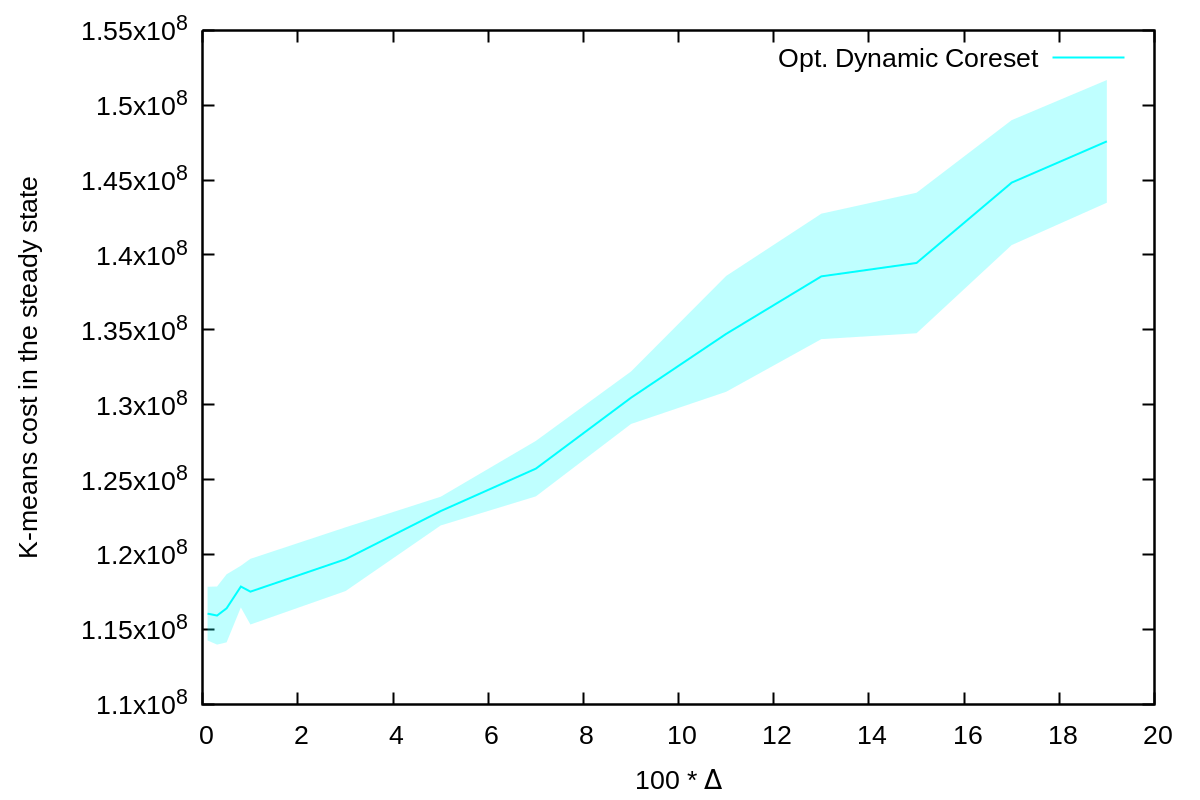}
    \caption{Cost of the \kmeans~solution of the sliding window with ordered deletions}
    \end{subfigure}
    \hfill
    \begin{subfigure}[b]{0.49\textwidth}
    \centering
    \includegraphics[width=\textwidth]{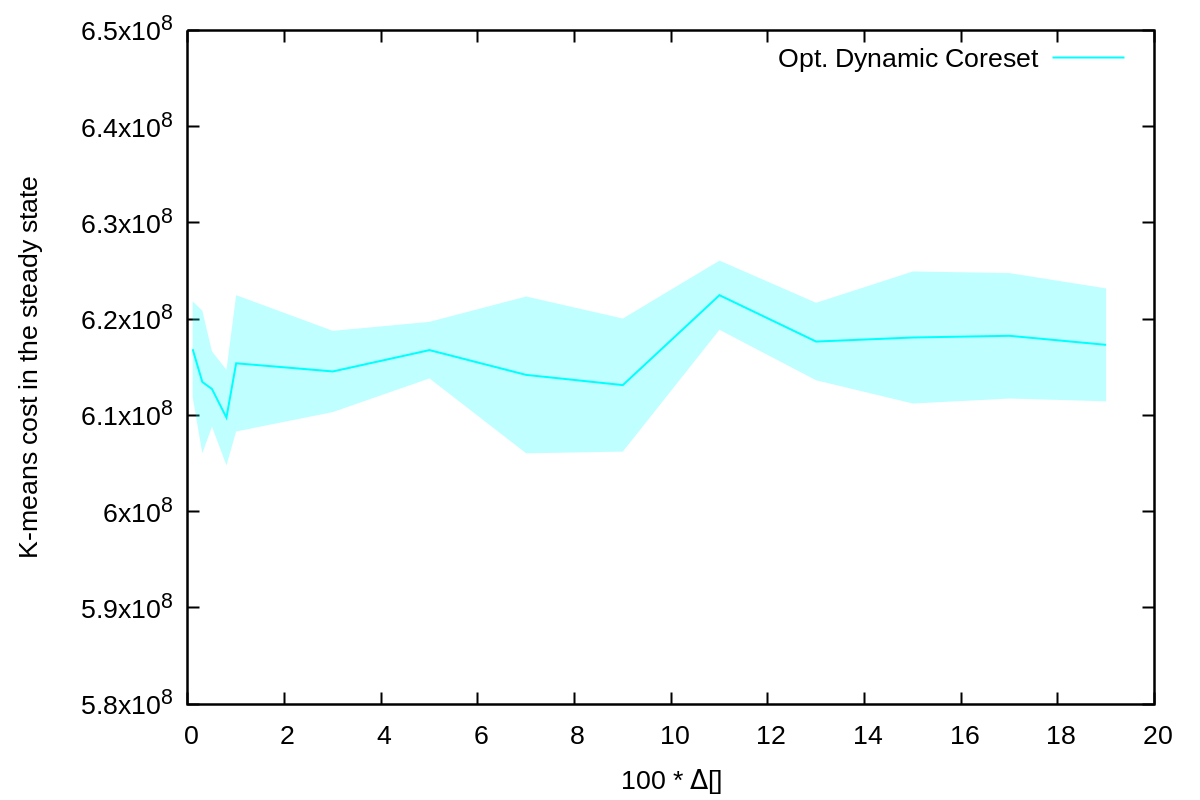}
    \caption{Cost of the \kmeans~solution of the sliding window with random deletions}
    \end{subfigure}
    \caption{\textbf{Impact of $\Delta$ on the remove optimization.} The Birch dataset with a sliding window of size $10000$ with the parameters given in \cref{tab:params} was used. Points were removed in the same order as inserted or in random order. The cost of the optimal \kmpp~solution on the coreset and running time is averaged over all steps in the sliding window.}
    \label{fig:ex_d}
\end{figure}

\FloatBarrier
\subsection{Shallow Tree.}\label{ap:shallow} 

\Cref{fig:ex_d} contains details of experiments to determine the best values for the remove cut-off $\Delta$ while \Cref{fig:ex_shallowTimeG} and \Cref{fig:ex_shallowTimeH} show the impact of the two parameters $g$ and $h$ on the shallow tree. 

The difference between $k=10$ and $k=50$ can be explained by the fact that our algorithm needs to sometimes merge nodes in the tree (to keep each leaf large enough), while the shallow tree does not. This feature could be incorporated in our algorithm as well. However, this is very specific to the sliding window case, where the size of the dataset stays fixed. 

\begin{figure}[h]
    \centering
    \includegraphics[width=\textwidth]{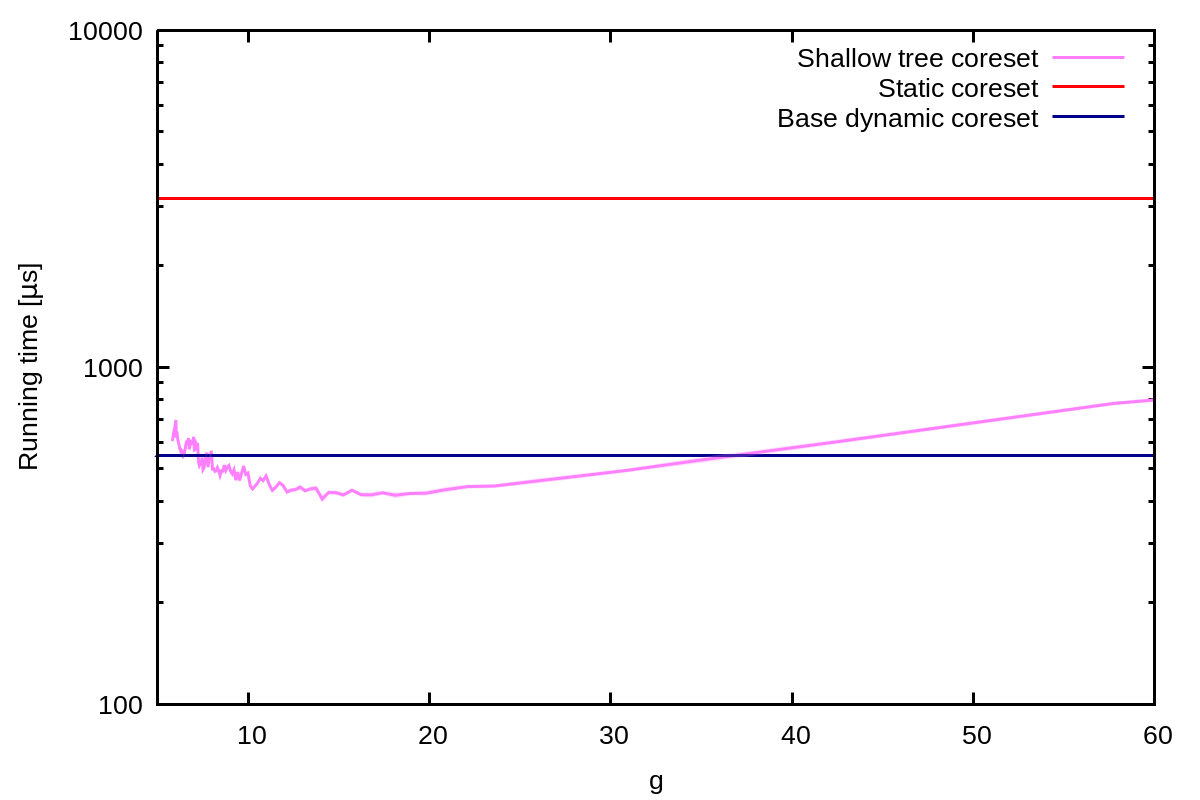}
    \caption{\textbf{Impact of \emph{g}} on the running time of the shallow tree of height one. All measurements were taken in a sliding window of size  $20000$ using parameters given in \cref{tab:params}.  The experiments were repeated five times and the median was used.}
    \label{fig:ex_shallowTimeG}
\end{figure}

\begin{figure}[h]
    \centering
    \begin{subfigure}[b]{0.49\textwidth}
    \includegraphics[width=\textwidth]{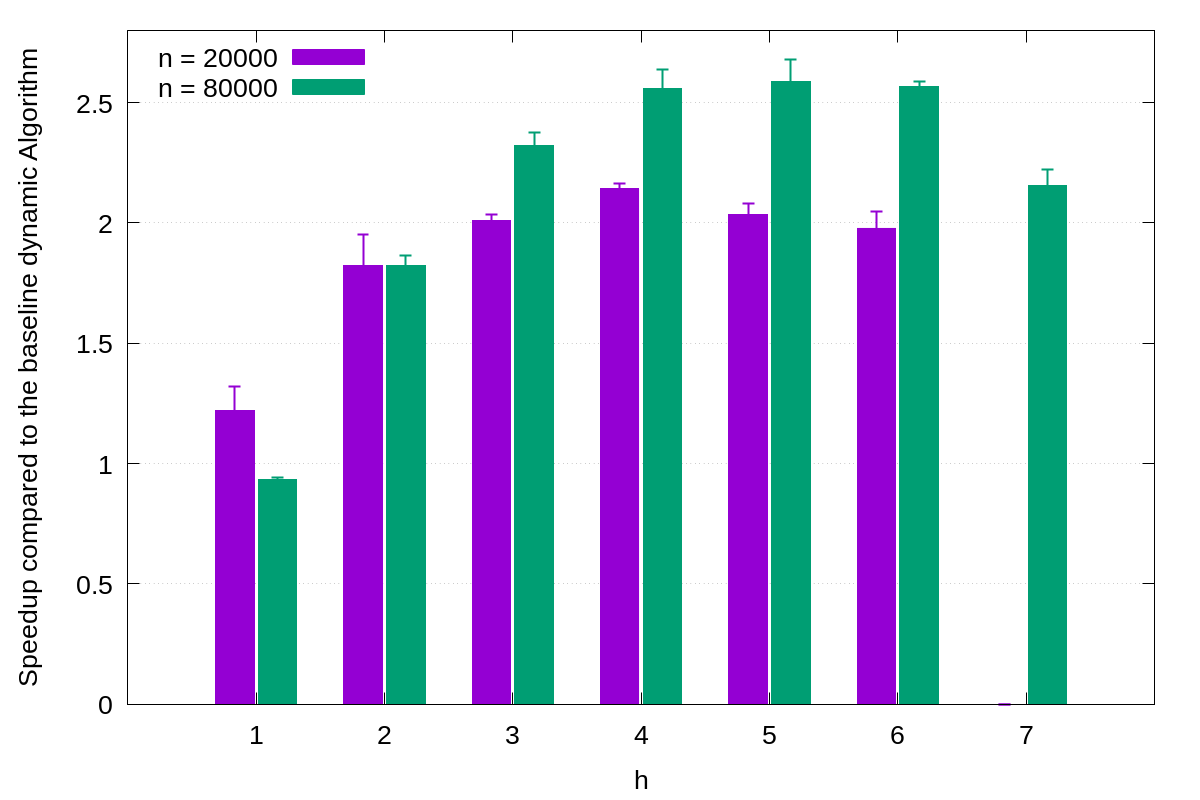}
    \caption{Impact of \emph{h} on the running time of the shallow tree of height one with $k = 10$ and $s = 50$.}
    \end{subfigure} 
    \hfill
    \begin{subfigure}[b]{0.49\textwidth}
    \includegraphics[width=\textwidth]{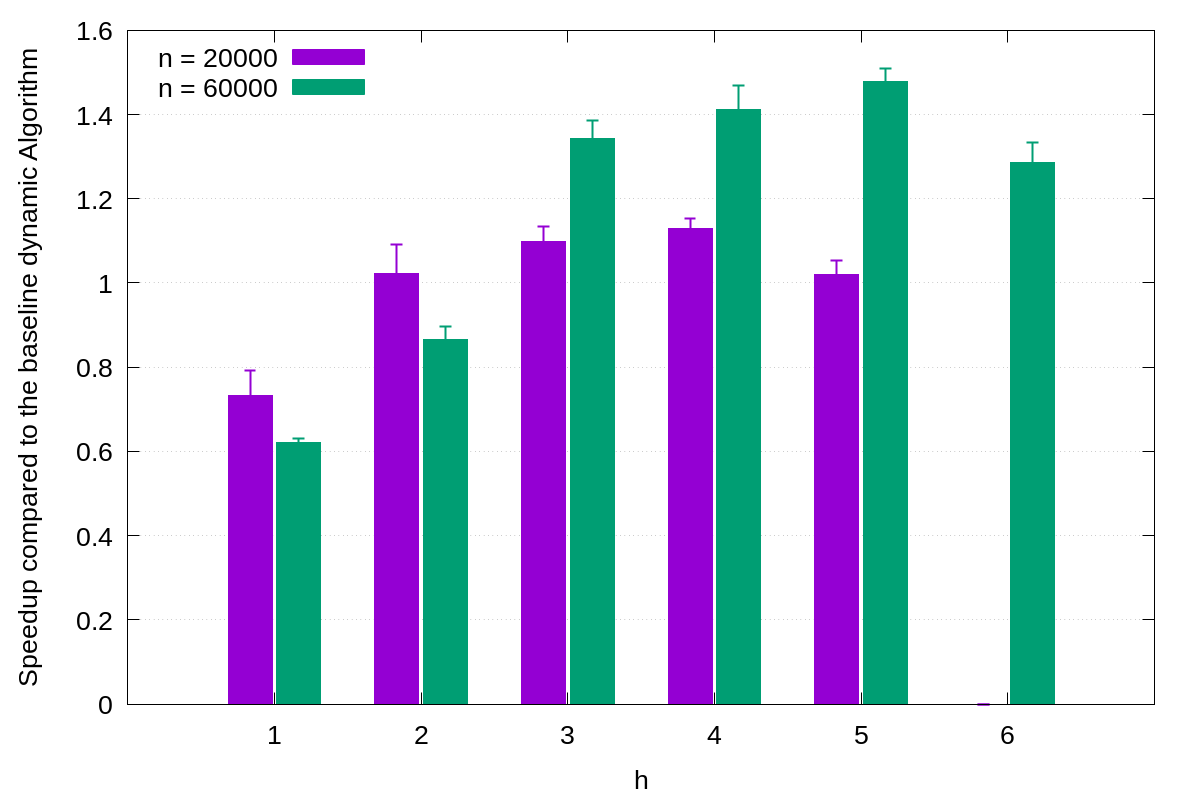}
    \caption{Impact of \emph{h} on the running time of the shallow tree of height one with $k = 50$ and $s = 90$.}
    \end{subfigure} 
    \label{fig:ex_shallowTimeH}
\caption{\textbf{Impact of \emph{h}} on the running time of the shallow tree with optimal $g$. All measurements were taken in a sliding window of different sizes using parameters given in \cref{tab:params}. 
}
\end{figure}

\FloatBarrier
\subsection{ILP-based algorithm.} \label{par:ex_ILP}

We tested the quality of this algorithm on a small dataset containing only insertions (see \Cref{tab:params}), running the ILP from scratch at every iteration. As its running time is by several orders of magnitude larger than any other algorithm, we introduced a time limit of $100.000$ times the running time of the dynamic coreset algorithm. This limit was already reached for very small input sizes of around $300$ points. See \Cref{fig:ex_ILP} for a comparison on the results up to this data size. The solutions found by the ILP have slightly better cost than the dynamic corset algorithm, but only when the ILP fully converges before the time limit. For larger inputs the solver could not produce solutions of a similar quality. 
ILP was therefore not further investigated. Instead, we used the cost achieved as $\kmpp$ algorithm as the benchmark for quality.

\begin{figure}[h]
    \centering
    \begin{subfigure}[b]{0.49\textwidth}
    \includegraphics[width=\textwidth]{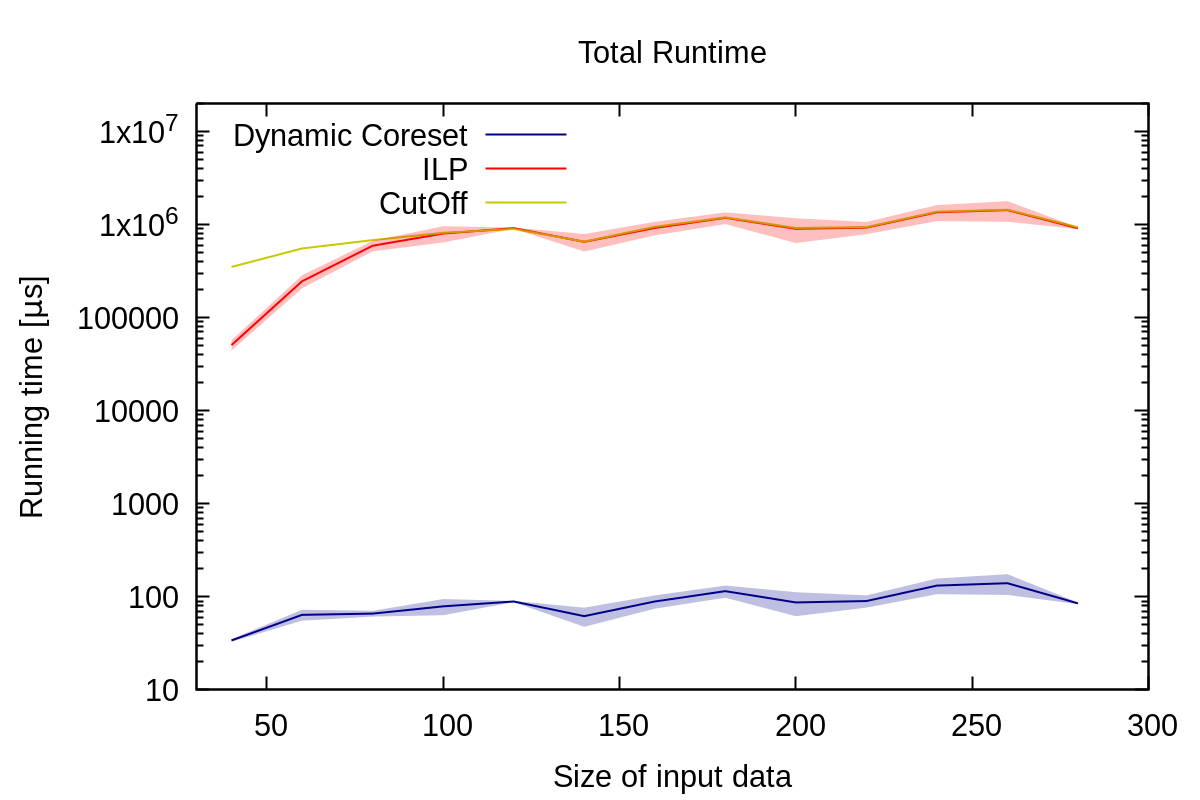}
    \caption{Running time comparison}
    \end{subfigure}
    \hfill
    \begin{subfigure}[b]{0.49\textwidth}
    \centering
    \includegraphics[width=\textwidth]{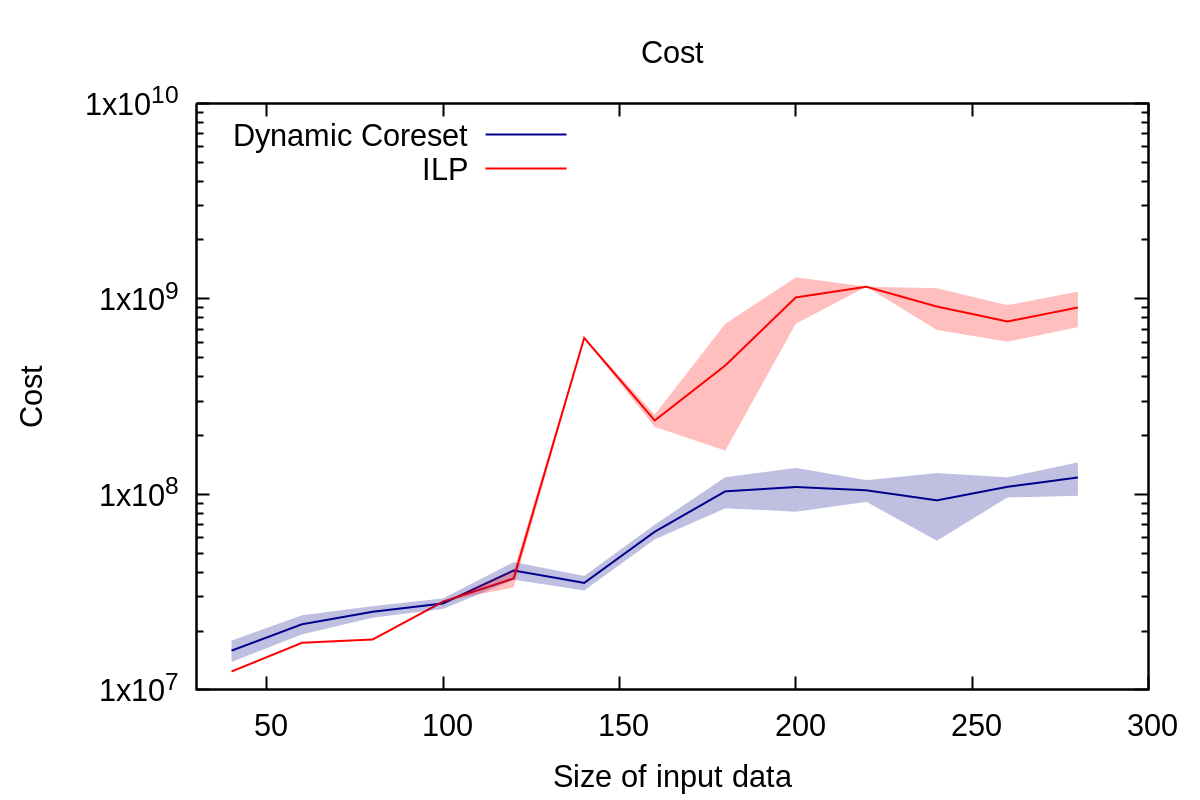}
    \caption{Cost of the \kmeans~solutions on different coresets }
    \end{subfigure}
    \caption{\textbf{Comparison with ILP.} An insertion only data set from the first 300 points of the Birch-snake dataset was used with parameters given in \cref{tab:params}. The cutoff for the running time of the ILP was set to $1e5$ times the running time of the dynamic coreset. Measurements are averages over all update operations. The experiments were repeated three times and the median was used.}
    \label{fig:ex_ILP}
\end{figure}

\end{document}